\documentclass{amsart}

\usepackage{amssymb, amsmath, amsthm, mathtools}

\usepackage{graphicx}
\usepackage{dcolumn}
\usepackage{bm}
\usepackage{color}
\usepackage{float}
\usepackage{tikz-cd}

\usepackage{hyperref}

\newtheorem{theorem}{Theorem}[section]
\newtheorem{proposition}[theorem]{Proposition}
\newtheorem{lemma}[theorem]{Lemma}
\newtheorem{corollary}[theorem]{Corollary}

\newtheorem{example}[theorem]{Example}

\newcommand{\extd}{\mathrm{d}}
\newcommand{\id}{\mathrm{id}}

\newcommand{\tens}{\otimes}

\newcommand{\del}{\partial}

\newcommand{\Vol}{\mathrm{Vol}}
\renewcommand{\>}{\rangle}
\newcommand{\<}{\langle}
\newcommand{\eps}{\epsilon}

\newcommand{\C}{\mathbb{C}}
\newcommand{\N}{\mathbb{N}}
\newcommand{\R}{\mathbb{R}}
\newcommand{\Z}{\mathbb{Z}}

\newcommand{\CJ}{\mathcal{J}}

\newcommand{\CC}{\mathcal{C}}

\newcommand{\CL}{\mathcal{L}}

\newcommand{\cg}{\mathfrak{g}}

\renewcommand{\imath}{{\mathfrak{i}}}
\renewcommand{\top}{{\mathrm{top}}}

\begin{document}

\title{Quantum variational calculus on a lattice}


\address{
School of Mathematical Sciences\\
Queen Mary University of London\\  Mile End Rd, London E1 4NS, UK}

\email{s.majid@qmul.ac.uk, f.castelasimao@qmul.ac.uk}

\thanks{SM was supported by Leverhulme Trust project grant RPG-2024-177}
\author{Shahn Majid and Francisco Sim\~ao}

\date{\today\ Ver 1.1}

\begin{abstract} We solve the long-standing problem of variational calculus on a noncommutative space or spacetime for a significant class of models with trivial jet bundle. Our approach entails a  quantum version of the Anderson variational double complex $\Omega(J^\infty)$ and includes Euler-Lagrange equations and a partial Noether's theorem. We show in detail how this works for a free field on a $\Z^m$ lattice regarded as a discrete noncommutative geometry, obtaining the Klein-Gordon equation for a scalar field, including with a general metric and gauge field background, as the Euler-Lagrange equations of motion for an action.  In the case of a flat metric we also obtain an exactly on-shell conserved stress-energy tensor and Noether charges for a scalar field on the lattice and modified energy-momentum relations. \end{abstract}

\maketitle

\section{Introduction}

It is widely believed that spacetime at the level of the Planck scale is not in fact a continuum due to quantum gravity corrections. What exactly it is is unclear but one plausible `quantum spacetime hypothesis' that has garnered much attention since the early works such as~\cite{Ma:pla,MaRue, Dop,Hoo,Sny} is that it is better modelled by quantum or noncommutative geometry where coordinates need not commute among themselves. It has also become clear that this more general vision of geometry also includes the case of a discrete spacetime, such as a lattice, which turns out to be an exact noncommutative geometry in which differential forms and functions do not commute (this is because a finite difference differential is a bilocal object with a source and target of each `step'~\cite{BegMa}). We recently revisited gauge theory on a lattice from this point of view, with strikingly different results that go beyond the usual lattice gauge theory~\cite{MaSim2}.

What has been missing for over 30 years, however,  is the single most important thing needed for a convincing derivation of physics on a noncommutative space or spacetime, including on a lattice as an exact geometry. Namely, a formulation of variational calculus whereby an action functional as used in a path integral formulation of the quantum field theory on such spaces is connected to equations of motion on such spaces and to conserved quantities. Without this, current attempts to formulate physics, while promising, lack a complete picture even at the level of classical field theory let alone quantum field theory. In the present work, we propose a solution to this long-standing problem for a limited class of noncommutative geometries but one that includes the lattice case.

Our approach is two-fold. First, we build on previous work~\cite{MaSim1,Flo} to construct a jet bundle with sections $\CJ_E^\infty$ over a (possibly noncommutative) coordinate algebra $A$, where $E$ is a vector bundle. We focus  on the case where the latter is trivial with real sections $A\tens\R$, so that the jet bundle has sections $\CJ^\infty$ over $A$ as in~\cite{MaSim1}. Next, we suppose that the jet bundle itself is trivial and that the fibre is classical, so there is an algebra which we denote $C(J^\infty):=A\tens C(\R^\infty)$ where $C(\R^\infty)$ is a suitable class of functions (we will use polynomials in an infinite number of variables). In the classical case $A=C^\infty(M)$ for a smooth manifold $M$ but we will also allow that $M$ could be discrete and in this case we just write $C(M)$ for a suitable class of functions, such as with compact support. In such cases one has an actual space $J^\infty=M\times\R^\infty$ where the $\R^\infty$ encodes the higher tangents as the jet bundle fibre. This will be our main case in the present work, with noncommutative $A$ considered elsewhere. We still need $A$ to be equipped with a space of differential forms  $\Omega^1(A)$ as common to all main approaches to noncommutative differential geometry including~\cite{Con} (where it is derived from a Dirac operator), but for the discrete case this exactly means a graph structure on $M$, where the 1-forms are labelled by the arrows.

Secondly, which is the main result at a technical level,  we need to construct a noncommutative exterior algebra which we denote $\Omega(J^\infty)$ on $C(J^\infty)$ in such a way as to form a double complex with vertical and horizontal exterior derivatives matching the classical case. This then allows us to follow the ideas of Anderson, Zuckerman~\cite{And, Zuc} and others in the classical theory of variational calculus to write down Euler-Lagrange equations of motion associated to a choice of  Lagrangian formulated in terms of $\Omega(J^\infty)$, and (here we have only partial results) a Noether's theorem associated to symmetries. Our constructions are necessarily quite mathematical but our final results are self-contained and explicit. As far as we know even our conserved energy
\begin{equation*}
    \label{Ephi} E[\phi]=- {1\over 2}(\del_+\phi)(\del_-\phi)+{m^2 \over 2}\phi^2
\end{equation*}
of a free scalar field of mass $m$ on a lattice line $\Z$ appears to be new (this is the simplest case of Corollary~\ref{corQ}). Here, $(\del_\pm \phi)(i)=\phi(i\pm1)-\phi(i)$ and the claim is that $E[\phi]$ is constant on $\Z$ if $\phi$ obeys the discrete wave equation $(\Delta_\Z+m^2)\phi=0$, where $\Delta_\Z\phi=\del_++\del_-$ is the usual discrete Laplacian. There are bounded oscillatory  solutions here for all real values of $E[\phi]$ and $m<2$ relative to the discretization scale but real solutions only for $E[\phi]>0$. We will look at this look this from the point of view of a discrete time harmonic oscillator as a classical particle $q\colon \Z \to \R$ with frequency $\omega$ playing the role of $m$, see Figure~\ref{fig1}, and from a scalar field theory, see Figures~\ref{fig:DispEnMomEuc}-\ref{fig:VelMin}. On a literature search, we noted~\cite{Car}, but this does not appear to have an exactly conserved stress tensor.

An outline of the paper is as follows. Since the paper involves a number of less familiar methods, including  the variational double complex, we will build up the theory in several layers starting with a preliminary Section~\ref{secM} explaining in detail how everything proceeds in the case of classical field theory on a classical manifold $M$. Section~\ref{secZ} then gives the discrete version where $M$ is replaced by a lattice line $\Z$, as the easiest  case of our theory. Section~\ref{secZm} then covers $\Z^m$ and $\Z^{1,m-1}$ (i.e. Euclidean and Minkowski lattice cases) and more generally for the `spacetime' any discrete group $G$ provided we can construct the jet bundle. This includes all Abelian groups but in principle also some nonAbelian groups such as the group $S_3$ of permutations of three elements~\cite{MaSim1}. Section~\ref{sec1+1} focusses on  (1+1)-dimensional lattice scalar fields in both the Euclidean and Minkowski cases, showing  modified energy-momentum relations. For example, in the Minkowski case we have 
 \[ \mathcal{E}^2+ \left(\frac{m^2}{2} - \sqrt{1-\mathcal{P}^2}\right)^2-1=0\]
for the solution that is close to the continuum case. Note that even in the continuum~\cite{Del},  classical variational calculus does not lead to the familiar form in quantum field theory as this requires quantisation, but one can get to the familiar form assuming an energy density per quantum. We used a similar strategy in the lattice case pending a treatment of quantum field theory on the lattice. Section~\ref{secO} looks at scalar fields with a general metric and/or $U(1)$ gauge field on the lattice as background, using methods of quantum Riemannian geometry~\cite{BegMa} and lattice gauge theory~\cite{MaSim2}. We end with some concluding remarks in Section~\ref{secrem} about further work.

\section{Algebraic classical variational calculus on $M$. }\label{secM}

In this preliminary section, we will recap how the abstract theory of calculus of variations appears in the case of a classical manifold. This theory is known cf~\cite{And,Zuc} but we need to recall it in an algebraic form that we can then `quantize'. This also identifies all the geometric ingredients that we need, some of which (notably Lie derivative and interior products) are less clear for a general noncommutative geometry, but clear enough in a variety of examples. Another feature is that we work at the polynomial level at the level of the jet bundle coordinates.

\subsection{Jet bundle and variational bicomplex}

We focus on $E=M \times \R \to M$ as the bundle for matter fields so that $\Gamma(E)=C^\infty(M)=F$ is the `space of matter fields' as a linear space, where a function $\phi\in F$ is viewed as a map $M \to E$ sending $x \mapsto (x, \phi(x))$.  The jet bundle here is trivial and this allows us to write the jet prolongation map $j_\infty\colon F\to \Gamma(J^\infty)$ explicitly as 
\[ J^\infty= M \times \R^\N=\{(x, u,u_{i},u_{ij},\cdots)\},\] \[j_\infty(\phi)(x)=(x, \phi(x), \del_{i} \phi(x),\del_i \del_j \phi(x),\cdots)=\{(x, \del_{I}\phi)\}\]
where the indices $i$ run over the dimension of $M$, $I= \{i_1,\cdots, i_n\}$ are multi-indices and $\del_I = \del_{i_1}\cdots \del_{i_n}$. We {\em also} regard $x^i, u_I$ tautologically as coordinates on $J^\infty$ so that $u_I(j_\infty(\phi)(x))=(\del_I \phi)(x)$. We define the evaluation map 
\[ e_\infty\colon M\times F\to J^\infty,\quad (x,\phi)\mapsto j_\infty(\phi)(x)\] 
which we assume in an appropriate context is surjective and smooth, so that we get a pull-back inclusion at the level of the exterior algebra $\Omega$,
\[ e_\infty^* \colon \Omega(J^\infty){ \hookrightarrow} \Omega(M)\underline\tens \Omega(F).\]
The key idea here is that the right-hand side, as a graded tensor product, is automatically a double complex with usual horizontal and vertical differentials $\extd_M,\extd_F$, respectively, which then induces a double complex structure on $\Omega(J^\infty)$ with $\extd_H,\extd_V$ corresponding to these.  

To see what this looks like, we assume that $\Omega^1(J^\infty)=C^\infty(J^\infty) \{ \extd x^i, \extd u_I \}$ and that $\extd_H,\extd_V$ are induced as
\[ e_\infty^*(\Phi)=\Phi(j_\infty(\cdot)(\cdot)),\qquad  e_\infty^*( \extd_H\Phi)=\extd_M e_\infty^*(\Phi),\qquad e_\infty^*( \extd_V\Phi)=\extd_F e_\infty^*(\Phi)\]
for $\Phi \in C^\infty(J^\infty)$, so in particular
\[ e_\infty^*( \extd_H\Phi)={\del_i}\big(\Phi (j_\infty(\cdot)(\cdot))\big)\extd x^i=\left(\left({\del_i \Phi}\right) (j_\infty(\cdot)(\cdot))+ \left({\del \Phi\over\del u_I}\right)(j_\infty(\cdot)(\cdot)) \del_{i I}(\cdot)(\cdot)\right)\extd x^i,\]
where we evaluated at $\phi\in F$ and use the chain rule. However, we can compute
\[ e_\infty^*(\extd x^i)=\extd e_\infty^*(x^i)=\extd x^i,\quad e_\infty^*(\extd u_I)=\extd \del_I(\cdot)(\cdot)=(\del_{iI}(\cdot))(\cdot) \extd x^i +\extd_F \del_{I}(\cdot)(\cdot).\]
Then
\[e_\infty^*\left((\del_i \Phi )\extd x^i\right)=e_\infty^*\left(\del_i \Phi\right)\extd x^i ={\del_i \Phi}(j_\infty(\cdot)(\cdot))\extd x^i,
\]
\[e_\infty^*\left({\del\Phi\over\del u_I} u_{iI} \extd x^i\right)={\del\Phi\over\del u_I}(j_\infty(\cdot)(\cdot))\del_{iI} (\cdot)(\cdot)\extd x^i.
\]
Comparing with the above, we see that 
\[ \extd_H \Phi=(D_i \Phi)\extd x^i,\qquad D_i\Phi:={\del_i \Phi}+ \sum_I{\del\Phi\over\del u_I}u_{iI}.\]
Similarly, 
\[ e_\infty^*(\extd_V\Phi)=\extd_F \Phi(j_\infty(\cdot)(\cdot))={\del\Phi\over\del u_I}(j_\infty(\cdot)(\cdot))\extd_F\del_{I}(\cdot)(\cdot)=e_\infty^*\left(\left({\del\Phi\over\del u_I}\right)(\extd u_I- u_{iI}\extd x^i)\right)\]
so that 
\[ \extd_V \Phi=\sum_I {\del \Phi\over\del u_I}(\extd u_I - u_{iI}\extd x^i).\]
Here, $\extd_V u_I = \extd u - u_{iI}\extd x^i$ are called {\em contact 1-forms}. As a check, we see that $\extd =\extd_H + \extd_V$ holds on $C^\infty(J^\infty)$. This induces a factorisation of the deRham complex into a double complex
\[
\Omega(J^\infty) = \bigoplus_{p,q\geq 0} \Omega^{p,q}(J^\infty)
\]
where $p,q$ are the horizontal and vertical degrees respectively and elements of degree $p,q$ are of the form
\[  \sum \phi_0 (\extd_H \phi_1)\cdots (\extd_H \phi_p)\wedge (\extd_V\psi_1)\wedge\cdots\wedge(\extd_V \psi_q)\]
for $\phi_i, \psi_i \in C^\infty(J^\infty)$.
 
Note that it does not matter too much what $\extd_F \del_I$ is as it cancels in the calculation, but we can let $\phi=\int \delta_y \phi^y \extd y$ be an expansion in a delta-function $\delta_y(x)$ basis and coordinatise  $F$ by the pointwise values $\{\phi^y\ |\ y\in \R$\}. Then if $\Psi\in C^\infty(F)$, $\extd_F\Psi=\int\extd y {\del \Psi \over\ \del \phi^y}\extd \phi^y$ and hence in particular,
\[ e_\infty^*(\extd u_I)=(\del_{iI}(\cdot))(\cdot)\extd x^i +\int\extd y \left({\del\over\del \phi^y}\del_{I}(\cdot)\right)(\cdot)\extd \phi^y=(\del_{iI}(\cdot))(\cdot)\extd x^i +\int\extd y (\del_{I}\delta_y)(\cdot)\extd \phi^y\]	
where we assume the $\delta_y$ functions have been smoothed so that we can differentiate them as functions of $x$. The second term is constant on $F$ and we have also suppressed that it is a function on $M$. We do not have to chose a $\delta$-function basis and more reasonable here would be a plane-wave basis i.e. to coordinatise $\phi\in F$ by its Fourier coefficients. These matters can be made more precise by usual methods in mathematical physics but this is not needed for our purposes.

\subsection{Euler-Lagrange equations}

In order to use the above setting to derive the Euler-Lagrange (EL) equations, let $L \Vol \in \Omega^{\top,0}(J^\infty)$, where $L$ is a first-order Lagrangian $L = L(u,u_i) \in C^\infty(J^\infty)$. The action is then defined as
\[
S[\phi] := \int_M e_\infty^*(L \Vol )(x,\phi)
\]
and its variation reads
\[
\extd_F S[\phi] = \int_M e_\infty^*(\extd_V L \wedge \Vol)(x,\phi).
\]
Computing the RHS gives $\extd_V L \wedge \Vol = EL - \extd_H \Theta$ for certain $EL\in \Omega^{\top,1}(J^\infty)$ the EL form and $\Theta \in \Omega^{\top-1,1}(J^\infty)$ the boundary term~\cite{Zuc}. Explicitly, we have
\[
    \extd_V L \wedge \Vol
    = \left(\frac{\del L}{\del u} \extd_V u + \frac{\del L}{\del u_i} \extd_V u_i \right)\wedge \Vol
\]
but $\extd_H \extd_V u = - \extd_V \extd_H u = -\extd_V u \wedge \extd x^i$. Setting $\Vol_i \coloneqq \iota_{\del_i}\Vol$, we write
\begin{align*}
    \extd_V L \wedge \Vol &= \frac{\del L}{\del u} \extd_V u \wedge \Vol - \frac{\del L}{\del u_1} \extd_H \extd_V u \wedge \Vol_i
    \\
    &=  \frac{\del L}{\del u} \extd_V u \wedge \Vol - \extd_H\left(\frac{\del L}{\del u_i} \extd_V u \wedge \Vol_i\right) + \extd_H\left(\frac{\del L}{\del u_i}\right) \extd_V u \wedge \Vol_i \\
    & = \left(\frac{\del L}{\del u} - D_i \left(\frac{\del L}{\del u_i}\right)\right) \extd_V u \wedge \Vol
    - \extd_H\left(\frac{\del L}{\del u_i} \extd_V u \wedge \Vol_i \right).
\end{align*}
and we find the EL and boundary term to be
\begin{equation}
    \label{eq:ClELBound}
    EL = \left(\frac{\del L}{\del u} - D_i \left(\frac{\del L}{\del u_i}\right)\right) \extd_V u \wedge \Vol,
    \qquad
    \Theta = \frac{\del L}{\del u_i} \extd_V u \wedge \Vol_i.
\end{equation}
This result can be generalised to Lagrangians including higher derivative terms $u_I$~\cite{And}, but 1st order Lagrangians that are functions of one derivative are sufficient for us. We say that a field $\phi \in F$ satisfies the equations of motion, or is `on-shell', if
\begin{equation}
    \label{eq:EOMform}
    e^*_\infty(EL)(\cdot,\phi) = 0,
\end{equation}
or equivalently
\[
    e^*_\infty\left(\frac{\del L}{\del u} - D_i \left(\frac{\del L}{\del u_i}\right)\right)(\cdot,\phi) = 0.
\]

\begin{example}
    \label{ex:ClMech}
    \rm We consider (non-relativistic) classical mechanics in 1 spatial dimension by taking the base $M = \R$ as the time dimension. The fibre $\R$ of the bundle $\R\times \R$ then represents the space dimension and a particle trajectory is a section of this as specified by a function $\phi=q\colon \R \to \R$. The action is
    \[
        S[q] = \int \left({m \over 2} \dot q^2 - V(q)\right) \extd t
    \]
    with $\dot q := \del_t q$ and $V(q)$ a potential term, resulting in the Lagrangian $L = {m^2 \over 2} u^2_t- V(u)$. The EL form is $$EL = \left(-{\del V(u)\over \del u} - m  u_{tt}\right) \, d_V u \wedge \extd t,$$ which recovers Newton's equation for a particle in a potential $$m \ddot q = -{\del V(q)\over \del q}.$$ The boundary form is $\Theta = m u_t \extd_V u$. 
\end{example}

\begin{example} \rm\label{ELR}
Consider free scalar field theory on the base $M = \R^n$ with the Euclidean metric or $M = \R^{1,n-1}$ with the Minkowski metric, both denoted by $g$, then the action is
\[
    S[\phi] = \frac{1}{2} \int (g^{ij}\del_i \phi \del_j \phi  - m^2 \phi^2)\extd^n x
\]
corresponding to $L = \frac{1}{2} (g^{ij}u_iu_j - m^2 u^2)$. The EL form is $EL = -(m^2 u + g^{ij}u_{ij}) \extd_V u \wedge \extd^n x$ which recovers the Klein-Gordon equation $(g^{ij}\del_i\del_j + m^2)\phi = 0$ for the metrics in question. The boundary form is $\Theta =  g^{ij} u_j \extd_V u \wedge \Vol_i $.
\end{example}

\subsection{Symmetries and Noether's Theorem}
\label{sec:SymNoeCl}

We start with some general background following~\cite{And} for a classical manifold as base $M$ and $E\to M$ a vector bundle. Symmetries at the infinitesimal level are given by a vector field on the total space of $E$, in local coordinates
\[
    X_E = X^i \del_i + X^a \frac{\del}{\del u^a}
\]
where the index $a$ runs over the fiber dimension in $E$. Such vector fields have a canonical prolongation to vector field on $J^\infty$~\cite{And} given by  
\[ X_\infty = X^i \del_i + X^a_I \frac{\del}{\del u^a_I}, \qquad 
    X_I^a = D_I(X^a - X^i u^a_i) + X^i u^a_{i I},
\]
where $D_I = D_{i_1}\cdots D_{i_k}$ for the multiindex $I = \{i_1,\cdots,i_k\}$. This can be split into horizontal and vertical components $X_\infty = X_H + X_V$ with respect to the contact structure, 
\[
    X_H = X^i D_i, \qquad
    X_V = D_I (X^a - X^i u^a_i) \frac{\del}{\del u^a_I},
\]
such that $[\iota_{X_H},\extd_V] = [\iota_{X_V},\extd_H] = 0$ and
\[
    \iota_{X_H} \extd x^i = X^i,
    \qquad
    \iota_{X_V} \extd x^i = 0,
    \qquad
    \iota_{X_H} \extd u^a_I = 0,
    \qquad
    \iota_{X_V} \extd_V u^a_I = D_I (X^a - X^i u^a_i).
\]
Following~\cite[p.165ff]{Del}, one says that $X_E$ is a {\em symmetry} if there is a form $\sigma_X \in \Omega^{\top-1,0}(J^\infty)$ such that
\[
    \CL_{X_\infty} (L\Vol) = \extd_H \sigma_X.
\]
Then Noether's theorem states that
\begin{theorem}cf.~\cite[p.165ff]{Del} Given a symmetry $(X_E,\sigma_X)$, there is an associated (on-shell) {\em conserved current} $j_X := \sigma_X - \iota_{X_H} (L\Vol) - \iota_{X_V} \Theta \in \Omega^{\top-1,0}(J^\infty)$. Then
\[
    \extd_H j_X = \iota_{X_V} EL.
\]
\end{theorem}

\begin{proof}
This is not new but as a model for later, we recall that the proof is to compute
\[
    \extd_H j_X = \CL_{X_\infty} (L\Vol) - \CL_{X_H}(L\Vol) - \iota_{X_V} \extd_H \Theta = \CL_{X_V} (L\Vol) - \iota_{X_V} \extd_H \Theta 
    = \iota_{X_V} EL
\] 
where we have used $[\iota_{X_H},\extd_V] = [\iota_{X_V},\extd_H] = 0$, $ \CL_{X} = \CL_{X_H} +  \CL_{X_V}$ and $\CL_{X_V} (L\Vol) = \iota_{X_V}\extd_V (L\Vol) = \iota_{X_V} EL - \iota_{X_V} \extd_H \Theta$.
\end{proof}

Here the current $j_X$ is conserved `on-shell' in the sense that it is $\extd_H$-closed up to the $EL$ form, i.e. $e^*_\infty(\extd_H j_X)(\cdot,\phi) = 0$ for fields $\phi$ which satisfy equation~\eqref{eq:EOMform}. To consider charges that are conserved  in a time direction, consider $M = \R \times \Sigma$ for some codimension 1 submanifold $\Sigma$ without boundary, with $\R$ denoting the time direction. In this setting, the coordinate on $\R$ will be denoted with the index $0$ and the coordinates on $\Sigma$ with the index $i$, so we write $j_X = j^0_X \Vol_0 + \sum^{\dim \Sigma}_{i=1} j^i_X \Vol_i$. We continue to focus on the Euclidean and Minkowski signatures and identify $\Vol_\Sigma$ with the volume form on $\Sigma$.

\begin{corollary}
    \label{cor:ClConCh}
The quantity $Q[\phi] = \int_\Sigma e_\infty^*(j^0_X)(\cdot,\phi) \Vol_\Sigma$ is conserved in time in the sense that $\del_0 Q$ vanishes on-shell. $Q$ is called a {\em conserved charge}. The quantities $\rho[\phi]\coloneqq e_\infty^*(j^0_X)(\cdot,\phi)$ and $J^i [\phi]\coloneqq e_\infty^*(j^i_X)(\cdot,\phi)$ are called charge and current densities, and satisfy the continuity equation $\del_0 \rho + \del_i J^i = 0$ when $\phi$ is on-shell.
\end{corollary}
\begin{proof}
We have $\extd_H j_X = (D_0 j^0_X + \sum^{\dim \Sigma}_{i=1} D_i j^i_X) \Vol$. Thus on-shell, we can compute
\[
    \del_0 Q[\phi] = - \int_\Sigma e_\infty^* (D_i j^i_X)(\cdot,\phi) \Vol_\Sigma = - \int_\Sigma \del_i(e_\infty^* (j^i_X) (\cdot,\phi)) \Vol_\Sigma  = 0
\]
as $\Sigma$ does not have a boundary. The continuity equation can be derived by directly applying $e^*_\infty$ to $\extd_H j_X$.
\end{proof}

After these general remarks, we now return to the example of the trivial bundle $E = M\times \R \to M$ to see how Noether's theorem and conserved charges work for a simple choice of Lagrangian. 

\begin{example}\rm
    For (non-relativistic) classical mechanics as in Example~\ref{ex:ClMech}, we have $M = \R$ (meaning that $\Sigma$ is a point). Here, the system is translation invariant, which can be encoded as $\iota_{X_\infty} (\extd t) = 1$, $\iota_{X_\infty} (\extd u_t) = 0$ so that $\iota_{X_V} (\extd_V u) = - u_{t}$, $\iota_{X_V} (\extd_V u_t) = - u_{tt}$, etc. This results in
    \[
        \sigma_X = 0, \qquad
        \iota_{X_H} (L\, \extd t) = {m\over 2} u^2_t - V(u),\qquad 
        \iota_{X_V} \Theta = - m u^2_t,
    \]
    leading to the conserved current $j_X = {m\over 2}u^2_t + V(u)$. The associated conserved charge then corresponds to the energy $Q = E = {m\over 2}\dot q^2 + V(q)$.
\end{example}

\begin{example} \rm
\label{ex:ClStressEnergy}
Consider $E = M\times \R\to  M$, with $M = \R^n$ with the Euclidean or Minkowski metric and a first-order Lagrangian $L = L(u,u_i)$ which is translation invariant. This symmetry can be encoded with $\iota_{X_\infty} (\extd x^i) = \epsilon^i$, $\iota_{X_\infty} (\extd u_I) = 0$ so that $\iota_{X_V} (\extd_V u_I) = - \epsilon^i u_{iI}$. Then
\begin{align*}
\CL_{X_\infty}(L \Vol) &= \CL_{X_H}(L \Vol) + \CL_{X_V}(L \Vol) = \extd_H(\iota_{X_H}(L\Vol)) + \iota_{X_V} \extd_V (L\Vol) \\
&= \extd_H L \wedge \epsilon^i \Vol_i -\epsilon^i
\left(\frac{\del L}{\del u} u_i + \frac{\del L}{\del u_j} u_{ij}\right) \Vol = 0
\end{align*}
so that $\sigma_X = 0$. Since $\Theta = \frac{\del L}{\del u_i} \extd_V u \wedge \Vol_i$ we have
\[
    j_X = \epsilon^i \left( \frac{\del L }{\del u_j} u_i - \delta^j_i L \right)  \Vol_j
\]
where we recognise the components of the stress-energy tensor as
\[
    T^i{}_j = \left(\frac{\del L }{\del u_i} u_j -\delta^i_j L \right).
\]
Divergence free then means
\[
    D_i T^i{}_j = D_i\left(\frac{\del L }{\del u_i}\right) u_j + \frac{\del L }{\del u_i} u_{ij} -  \frac{\del L }{\del u} u_j - \frac{\del L }{\del u_i} u_{ij} = -\left(\frac{\del L }{\del u}-D_i\left(\frac{\del L }{\del u_i}\right)\right)u_j, 
\]
which we see vanishes when the EL equations are fulfilled. 
For the specific case of the scalar field theory Lagrangian in Example~\ref{ELR}, we find for $T_{ij} = g_{ik}T^k{}_j$
\[
    T = \left(u_i u_j - \frac{1}{2}g_{ij} (g^{mn} u_m u_n - m^2 u^2)\right) \extd x^i \tens_{C^\infty(M)} \extd x^j \]
which on the level of the scalar field $\phi$ translates to
\[
    T[\phi]= \left(\del_i \phi \del_j \phi - \frac{1}{2}g_{ij} (g^{mn} \del_m \phi \del_n \phi - m^2 \phi^2)\right) \extd x^i \tens_{C^\infty(M)} \extd x^j
\]
where $T[\phi] \coloneqq e_\infty^*(T)(\cdot,\phi)$. It is easy to see  that $\nabla\cdot T=0$ if we use the equations of motion, where $\nabla\extd x=0$ is the flat connection on $\R$.

In the case $M = \R \times \Sigma$, the conserved charges are the energy and momenta corresponding to $T_{00}[\phi]$ and $T_{0i}[\phi]$ defined as
\begin{align*}
    \label{eq:ClEnergyScalarField}
    E[\phi] &= \int_\Sigma T_{00}[\phi] \, \extd^{n-1} x \\&= \int_\Sigma \left(\frac{1}{2}(\del_0 \phi)^2 -g_{00} \frac{1}{2} \sum^{\dim \Sigma}_{i,j=1} g^{ij}(\del_i \phi)(\del_j \phi) + g_{00} \frac{m^2}{2} \phi^2\right) \extd^{n-1} x,\nonumber
\end{align*}
\begin{equation*}
    \label{eq:ClMomentumScalarField}
    P_i[\phi] = \int_\Sigma T_{0i}[\phi] \, \extd^{n-1} x = \int_\Sigma \del_0 \phi \del_i \phi \, \extd^{n-1} x .
\end{equation*}
In particular, the energy for the Euclidean and Minkowski metrics, with $g_{00} = 1$, $g_{ii} = \pm 1$, respectively, is
\[
    E[\phi] = \int_\Sigma {1\over 2}\left((\del_0 \phi)^2 - g_{ii} \sum^{\dim \Sigma}_{i=1}(\del_i \phi)^2 + m^2\phi^2\right) \extd^{n-1} x,
\]
normally used in the Minkowski case.

Focusing now on the (1+1)-case for simplicity and choosing a plane wave solution $\phi(t,x) = A \cos(\omega t + g_{11} \kappa x)$ solving the Klein-Gordon equation $(\Delta + m^2)\phi = 0$ with the dispersion relations $m^2 = \omega^2 + g_{11}\kappa^2$, one recovers that the expected energy and momentum densities (the integrands of $E[\phi]$ and $P_x[\phi]$)
\[
    \mathcal{E}_{dens}[\phi] = {A^2 \over 2}\left(\omega^2 \sin^2(\omega t + g_{11} \kappa x) - g_{11} \kappa^2 \sin^2(\omega t + g_{11} \kappa x)+ m^2\cos^2(\omega t + g_{11} \kappa x)\right),
\]
\[
    \mathcal{P}_{x,dens}[\phi] = A^2 g_{11} \omega \kappa \sin^2(\omega t + g_{11} \kappa\cdot x).
\]
As is common in real scalar field theory, it is not directly clear that integrating these densities over $\Sigma$ will result in a time independent result. Instead, we can look at their spatial average over $[x_0,x_0 + 2\pi/\kappa]$ to find~\cite[p.194f]{Del}
\begin{align*}
    \< \mathcal{E}_{dens}[\phi]\> = \frac{\kappa}{2\pi} \int^{x_0 + 2\pi/\kappa}_{x_0} \mathcal{E}_{dens} \, \extd x = {A^2 \over 2} \omega^2,
\end{align*}
\begin{align*}
    \< \mathcal{P}_{x,dens}[\phi]\> = \frac{\kappa}{2\pi} \int^{x_0 + 2\pi/\kappa}_{x_0} \mathcal{P}_{x,dens} \, \extd x= g_{11}{A^2 \over 2} \omega \kappa,
\end{align*}
which makes the time independence of the energy and momenta clear. Using the dispersion relation we find that these densities are related as
\[
    \< \mathcal{E}_{dens}[\phi]\>^2 + g_{11} \< \mathcal{P}_{x,dens}[\phi]\>^2 = {A^2\over 2} \< \mathcal{E}_{dens}[\phi]\> \, m^2,
\]
and choosing the normalization $A = \sqrt{2/\< \mathcal{E}_{dens}[\phi]\>}$ of the field $\phi$ we then recover the familiar $\< \mathcal{E}_{dens}[\phi]\>^2 + g_{11} \< \mathcal{P}_{x,dens}[\phi]\>^2 = m^2$, 
indicating that the field $\phi$ has the average density of `one particle of mass $m$ per unit volume', with $\< \mathcal{E}_{dens}[\phi]\> = \omega$ and $\< \mathcal{P}_{x,dens}[\phi]\> = \kappa$. Note that the usual normalization $\sqrt{1/(2 \< \mathcal{E}_{dens}[\phi]\>)}$ found when quantizing the field via the CCR is recovered if we write the field $\phi(t,x) = (A/2) (e^{\imath (\omega t + g_{11} \kappa x)} + e^{-\imath (\omega t + g_{11} \kappa x)})$ in terms of exponentials.

\end{example}

\begin{example}\label{ex:ClU1} \rm For a different symmetry, consider a complex scalar field theory with \[S[\phi,\phi^*] = \frac{1}{2}\int_M (g^{ij}\del_i \phi^* \del_j \phi -m^2 |\phi|^2)\Vol\] which is invariant under the global $U(1)$ symmetry $\phi \mapsto e^{\imath \varphi}\phi$, $\phi^* \mapsto e^{-\imath \varphi}\phi^*$ for $\varphi \in [0,2\pi)$. In this example, $E$ is a trivial complex line bundle but we can also think of the fibre as $\R^2$ with a circular symmetry.
    
The jet bundle formalism is similar to before, but now with coordinates $u_I,\bar u_I$ to account for both fields, and the Lagrangian is 
\(
    L = \frac{1}{2} (g^{ij}\bar u_i u_j - m^2 \bar u u).
\)
The $U(1)$ symmetry can locally be encoded as
    \[
        \iota_{X_H} \extd x^i = 0, 
        \qquad 
        \iota_{X_V}(\extd_V u_I) = \imath \varphi u_{I}
        \qquad
        \iota_{X_V}(\extd_V \bar u_I) = -\imath \varphi \bar u_I.
    \]
Due to $\CL_{X_\infty}( L\Vol)= 0$, we have
    \[
        j_{U(1)} = -\iota_{X_V} \Theta = - \left(\frac{\del L}{\del u_i} \iota_{X_V} (\extd_V u)+ \frac{\del L}{\del \bar u_i} \iota_{X_V} (\extd_V \bar u)\right) \Vol_i
        = -\frac{\imath \varphi}{2} g^{im} (\bar u_m u - u_m \bar u)\Vol_i
    \]
    as expected. In terms of the field $\phi$ and its conjugate the conserved $U(1)$-current reads 
    \[
        j_{U(1)}[\phi,\phi^*] = -\frac{\imath \varphi}{2} g^{im} ( (\del_m \phi^*) \phi -  \phi^* (\del_m \phi)) \Vol_i.
    \]
    and for $M = \R \times \Sigma$, we find the conserved charge corresponding to the $U(1)$-symmetry to be
    \[
        Q_{U(1)}[\phi,\phi^*] = - \frac{\imath \varphi}{2}  g^{0m} \int_\Sigma ( (\del_m \phi^*) \phi -  \phi^* (\del_m \phi)) \extd^{n-1} x.
    \]
    with charge density
    \[
        J^m_{U(1)}[\phi,\phi^*] = - \frac{\imath \varphi}{2}  g^{mk} ((\del_k \phi^*) \phi -  \phi^* (\del_k \phi)).
    \]
    As in Example~\ref{ex:ClStressEnergy}, this Lagrangian is translation invariant for $g$ the Euclidean or Minkowski metrics, and a similar analysis results in the expressions
    \[
        T = (\bar u_i u_j -g_{ij} L) \extd x^i \tens_{C^\infty(M)} \extd x^j,
    \]
    \[
        T[\phi,\phi^*] = (\del_i \phi^* \del_j \phi -g_{ij} L) \extd x^i \tens_{C^\infty(M)} \extd x^j.
    \]
 We again focus on the (1+1)-dimensional case, where a plane wave solution $\phi(t,x) = A e^{- \imath (\omega t + g_{11} \kappa x)}$ of the Klein-Gordon equation satisfies the dispersion relation $m^2 = \omega^2 + g_{11}\kappa^2$. Then the energy, momentum, $U(1)$-charge and $U(1)$-current densities are recovered as expected by~\cite[p.194f]{Del}
    \[
        \mathcal{E}_{dens}[\phi,\phi^*] = T_{00}[\phi,\phi^*] = |A|^2 \omega^2 ,
    \]
    \[
        \mathcal{P}_{x,dens}[\phi,\phi^*] = T_{01}[\phi,\phi^*] =g_{11} |A|^2 \omega \kappa,
    \]
    \[
        \rho_{U(1),dens}[\phi,\phi^*] = |A|^2 \, \varphi \omega,
    \]
    \[
        J_{U(1),dens}[\phi,\phi^*] = |A|^2 \, \varphi \kappa.
    \]
 Note that $\mathcal{E}_{dens}[\phi,\phi^*]$, $\mathcal{P}_{x,dens}[\phi,\phi^*]$ are double the values of the real scalar field theory in Example~\ref{ex:ClStressEnergy} due to the existence of two different modes $\phi$, $\phi^*$ in this theory. In this case we find the relation
 \[
    \mathcal{E}_{dens}[\phi,\phi^*]^2 + g_{11} \mathcal{P}_{x,dens}[\phi,\phi^*]^2 = |A|^2 \mathcal{E}_{dens}[\phi,\phi^*] \, m^2,
\]
directly for the energy and momentum densities, which with $A =  \sqrt{1/\mathcal{E}_{dens}[\phi,\phi^*]}$ becomes
\[
    \mathcal{E}_{dens}[\phi,\phi^*]^2 + g_{11} \mathcal{P}_{x,dens}[\phi,\phi^*] = m^2.
\]
\end{example}
 
\section{Variational calculus on the integer lattice $\Z$}\label{secZ}

We now follow the exact same ideology but replacing the base manifold $M$ with $\Z$ as a discrete noncommutative geometry. This is an application of a general picture whereby $\Omega^1$ on a discrete set is precisely a graph, and here we consider $\Z$ as a graph, namely a one-dimensional lattice. Using that $\Z$ is a group, it is convenient to write the calculus in terms of a  basis of invariant 1-forms $e^a=e^\pm$ over the algebra of real valued functions $C(\Z)$. We will denote the algebra of functions in the base by $C(\Z)$ instead of using $\C(\Z)$ since we do not want to necessarily fix the field to be $\C$. In most of the chapter, we will work in fact with real-valued functions but we will also consider complex valued fields in some places, cf. Section~\ref{sec:U1symm}. We do not address the issue of $*$-structures which would be needed for a full understanding of the noncommutative geometry, since we do not yet have a theory of $*$-structures on jet bundles, as noted in~\cite{MaSim1}.

The exterior derivative $\extd\colon C(\Z)\to \Omega^1$ is then
\[ \extd_\Z f= \sum_a (\del_a f)e^a,\quad \del_\pm(f)(i)=f(i\pm 1)-f(i),\]
for $f\in C(\Z)$, where we adopt a suitable class of function on the lattice (say with bounded support to allow integration by parts). It is important to note that $e^\pm$ do not commute with functions. Rather $e^a f=R_a (f)e^a$ where $R_\pm(f)(i)=f(i\pm 1)$. It follows that $\del_\pm$ are not usual derivations but obey
\[ \del_\pm(fg)= (\del_\pm f)R_\pm (g)+f \del_\pm(g)\]
for all $f,g\in C(Z)$. Higher forms have a basis given by that of the Grassmann algebra on the $e^{a}$ (they anticommute among themselves) and $\extd=[\theta,\ \}$ where $\theta=\sum_a e^a$ and $[\ ,\ \}$ denotes a graded commutator.  This is how lattice geometry appears as an example of noncommutative geometry. It has been used in numerous works, see~\cite{Dim,Ma:par,BliMa,BegMa}, as an example of a general analysis for differentials on quantum groups~\cite{Wor}. In our case,  there is also a flat torsion-free connection $\nabla \colon \Omega^1 \to \Omega^1 \tens_{C(\Z)} \Omega^1$ characterised by $\nabla e^\pm=0$ (it is the QLC for the constant metric in the lattice in the sense of~\cite{BegMa}.) Notice that $\Omega^1(\Z)$ is 2-dimensional over the algebra because $\del_\pm$ are linearly independent as operators. For the same reason the top `volume' form ${\rm Vol}:=e^+\wedge e^-$ is a 2-form. We fix the above calculus $\Omega(\Z)$ and our first task is to extend it to a calculus $\Omega(J^\infty)$.

\subsection{Construction of the double complex $\Omega(J^\infty)$}
\label{secOmegaJ}
We start by considering $\Z \times \R \to \Z$, which in analogy to Example~\ref{ex:ClMech} represents a classical particle moving in one spacial dimension, where time is discretised. This is treated as a toy-model and later generalised to scalar field theory on the lattices $\Z^{m}$ and $\Z^{1,m-1}$ in Section~\ref{secZm}. The matter fields are sections of $\Z\times \R \to \Z$ that send $i\mapsto (i,\phi(i))$ for real-valued functions $\phi\in C(\Z)=F$ the field space. The symmetric tensors $\Omega_S$  are of the form
\[ \Omega_S=C(\Z)\{ 1, u_a e^a, u_{a_1a_2} e^{a_1}\tens e^{a_2}, \cdots \}=C(\Z)\{ u_I e^I \}\]
where $u_{a_1a_2}$ is symmetric and in the general case $I$ stands for $(a_1,\cdots,a_n)$ with the $u_{a_1\cdots a_n}$ totally symmetric tensors in the $a_i$ indices. When $n=0$ we just write $u$. These are specified by their values on on indices taken in (say) lexicographical ordering, for example in degree 2 we can specify $u_{++}, u_{+-}, u_{--}$. To be fully explicit, we let $I$ denote such an ordered set of indices, so the $u_I$ form a basis of such tensors (this is equivalent to saying there is a basis $e^I$ of symmetrized tensor products which in degree 2 say is  $e^+\tens e^+, e^+\tens e^-+e^-\tens e^+, e^-\tens e^-$). We take the $\{u_I\}$ as coordinates on the fibre, so
\[ J^\infty=\Z\times \R^\infty=\Z\times \R\times \R^2\times \R^3\times ...\]
with coordinates $(i, u_I)$. The jet prolongation map $j_\infty\colon C(\Z) \to \Gamma(J^\infty)$ is 
\[ j_\infty(\phi)= \phi+ \del_a \phi e^a+ \cdots=\sum_{I} \del_{I} \phi \, e^I\]
regarded as a function on $M=\Z$, i.e. the section itself sends $i$ to $(i, (\del_{I}\phi)(i))=:e_\infty(\phi)(i)$ where $e_\infty: \Z\times F\to J^\infty$ as before and pull back on this set map to (assumed an inclusion)  $\Omega(J^\infty)\subset \Omega(\Z)\underline\tens \Omega(F)$. 

The first thing that goes wrong is that one can no longer take a tensor product calculus $\Omega(J^\infty)=\Omega(\Z)\underline\tens \Omega(\R^\infty)$ with respect to which $e_\infty$ is differentiable in the sense that the pullback $e_\infty^*$ commutes with the differentials (this worked before due to everything having the same default differentiable structure on each copy of $\R$.) However, since $C(J^\infty)=C(\Z)[u_I]$ where we adjoin  commuting generators $\{u_I\}$ (i.e. tensor with the polynomial algebra $C(\R^\infty)$ with these generators), is a subalgebra of $C(\Z\times F)$ via $e_\infty^*$, we can still define $\Omega(J^\infty)$ as generated by this subalgebra and its inherited differentials, i.e. we must obtain a calculus just with certain noncommutation rules to be determined. In degrees 0,1 we have
\[ e_\infty^*(\Phi)(i,\phi)=\Phi(i, (\del_{I} \phi)(i)),\quad e_\infty^*(e^\pm)=\sum_i e_\infty^*(\delta_i \extd \delta_{i\pm1})=e^\pm,\] 
for $\Phi\in C(J^\infty)$, by the same arguments as before that if $\Phi=\psi\tens 1$ for $\psi\in C(Z)$ then $e_\infty^*(\Phi)(i,\phi)=(\psi\tens 1)(i, \del_I\phi)=\psi(i)$ so $e_\infty^*(\delta_j\extd\delta_{j+1})=\delta_j\extd \delta_{j+1}$ now viewed in $\Omega(\Z\times F)$. So we can identify the $\Omega(\Z)$ factors on the two sides. Next, 
\begin{align*} e_\infty^*(u_I)(i, \phi)&=u_I(i, (\del_{I} \phi)(i))=(\del_{I}\phi)(i),\\  e_\infty^*(\extd u_I)(i,\phi)&=(\extd e_\infty^*(u_I))(i,\phi)=  (\del_a \del_{I}\phi)(i)e^a+(\extd_F \del_{I})(\phi)(i).\end{align*}
It is not necessary but to be explicit, we can chose a basis $\{\delta_j\}$ of $F$ so that $\phi=\sum_j \delta_j\phi^j$ gives coordinates $\{\phi^j\}$ on $F$. Then $\Psi\in C^\infty(F)$ means functions $\Psi(\phi^j)$ and  $\extd \Psi= \sum_j {\del \Psi\over\del \phi^j}\extd \phi^j$. Then 
\[ e_\infty^*(\extd u_I)=(\del_{aI}(\cdot))(\cdot) e^a+ \sum_j (\del_I\delta_j)\extd \phi^j\]
where the dots indicate first to insert $\phi\in F$ then an element in $\Z$. Instead of proceeding with the chain rule as in the classical $M$ case, we work directly with the finite differences:
\begin{align*}
    e_\infty^*(\extd_H\Phi)&=\extd_M\Phi(j_\infty(\cdot)(\cdot))=\sum_a\left(e^a\Phi(\cdot, \del_I(\cdot))(\cdot ))- \Phi(\cdot, \del_I(\cdot)(\cdot))e^a\right) \\&= e_\infty^*\left(\sum_a\left(e^a\Phi- \Phi e^a\right)\right)
\end{align*} 
where sum over $a$, so 
\begin{equation}
    \label{dHPhi} 
    \extd_H \Phi=[\theta, \Phi],
\end{equation}
which is generally $\ne \extd_\Z \Phi$ precisely because $\Phi$ has $\R^\infty$ dependence and this is not the tensor product calculus. We see that $\extd_H=\extd_\Z$ on $C(\Z)$. Doing the same calculation for $\extd_V \Phi$ gives
\[ e_\infty^*(\extd_V \Phi)=\extd_F \Phi(\cdot,j_\infty(\cdot)(\cdot))={\del\Phi\over\del u_I}(\cdot, \del_I(\cdot)(\cdot))\sum_j \del_I\delta_j \extd \phi^j\]
where we keep $i$ fixed and vary $\phi$ in the $\delta_j$ direction for the coefficient of $\extd \phi^j$. Combining this with the above we can write this as
\begin{equation}
    \label{dVPhi} 
    \extd_V\Phi=\sum_I {\del \Phi\over\del u_I}\extd_V u_I,\quad \extd_V u_I=\extd u_I- [\theta, u_I].
\end{equation}
We see that $\extd_V=0$ on $C(\Z)$.

It remains to find commutation relations for working in $\Omega(J^\infty)$.  We calculate
\[ e_\infty^*(e^a u_I)=e^a \del_I(\cdot)(\cdot)= \big((\del_{aI})(\cdot)(\cdot)+ \del_I(\cdot)(\cdot)\big)e^a=e_\infty^*(u_{a I}+ u_I)e^a\]
so 
\begin{equation}\label{eurelns} e^a u_I= (u_I  + u_{a I} ) e^a,\end{equation}
where $a I$ denotes the standard lexicographic form of the indices with an extra $a$. For example, if $I=+-$ then $+I=++-$ and $-I=+--$. From this we conclude (summing $a$) that
\begin{equation}\label{dHu} \extd_H u_I =[\theta,u_I]= \sum_a u_{a I} e^a.
\end{equation}

Similarly, if $\phi\in C(\Z)\subset C(J^\infty)$  (i.e. constant on $\R^\infty$) then 
\[ e_\infty^*(\phi\extd u_I)=\phi\del_{aI}(\cdot)(\cdot) e^a+ \phi\sum_j (\del_I\delta_j)\extd \phi^j,\]
\[ e_\infty^*(\extd u_I \phi)=\sum_a\del_{aI}(\cdot)(\cdot)e^a\phi+ \sum_j (\del_I\delta_j)\extd \phi^j\phi=\sum_a\del_{aI}(\cdot)(\cdot)R_a(\phi)e^a + \phi\sum_j \left(\del_I\delta_j\right)\extd \phi^j\]
so we conclude that
\begin{equation}\label{duphi} (\extd u_I)\phi =\phi \extd u_I+ \sum_a (\del_a\phi)u_{a I} e^a\end{equation}
which indeed necessarily holds on applying $\extd$ to $[u_I,\phi]=0$. Similarly expanding
\begin{align*} e_\infty^*((\extd u_I )u_J)&=\sum_a\del_{aI}(\cdot)(\cdot)e^a \del_J(\cdot)(\cdot)=\sum_a \del_{aI}(\cdot)(\cdot) R_a\del_J(\cdot)(\cdot)e^a\\
&= e_\infty^*(u_J\extd u_I)+\sum_a \del_{aI}(\cdot)(\cdot) \del_{aJ}(\cdot)(\cdot)e^a\end{align*}
and recognising the last term, we obtain
\begin{equation}\label{duurelns}  (\extd u_I)u_J= u_J \extd u_I + \sum_a u_{a I} u_{a J}e^a.\end{equation}
We see that the $u_I$ do not inherit their classical commutative calculus, again due to $e^a$ not being central in the tensor product algebra.

The higher exterior algebra then follows.  By applying $\extd$ to the degree 1 relations, we have, using $\extd e^a=0$, 
\begin{equation}\label{edu} \{e^a, \extd u_I\}+ \extd u_{a I} \wedge e^a=0,\quad \{\extd \phi, \extd u_I\}+ \sum_a (\del_a\phi)\extd u_{a I}\wedge e^a+ \sum_{a,b}(\del_b\del_a\phi)u_{a I} e^b \wedge e^a=0,\end{equation}
where the second relation is redundant, i.e. can be proved by application of the other relations. We used that the $e^a$ anticommute to obtain it in this form. Next, using~\eqref{dHu} and  assuming $\extd_H$ is a derivation as it should be as inherited from $\Omega(\Z\times F)$ and that $\extd_H e^a=\extd e^a=0$, we have
\[ \extd_H\extd_H u_I=\extd_H\left(\sum_a u_{a I}e^a\right)=\sum_{a,b}u_{b a}e^b \wedge e^a=0\]
by anticommutativity  of the $e^a$. Hence these assumptions seem reasonable. Next, to impose $\extd_V\extd_H+\extd_H\extd_V=0$, as should also be inherited, is equivalent given the above to imposing  $\extd \extd_H+\extd_H\extd=0$, which on $u_I$ using~\eqref{edu} and~\eqref{dHu} reduces to imposing
\begin{equation}
    \label{dHinner} 
    \{\theta,\extd u_I\}= \extd_H \extd u_I
\end{equation}
where $\{\ ,\ \}$ denotes anticommutator. This seems reasonable and is consistent with $\extd_H=[\theta,\ \}$ being inner. This then implies that $\extd_V^2=0$ on $u_I$.  The $\extd_V,\extd_H$  then extend by the graded-derivation rules to all $\Omega(\Z\times \R^\infty)$ as one can check. Thus, we arrive at the following $\Omega(J^\infty)$.

\begin{theorem}\label{thm:Omega2rel} $C(J^\infty)$  generated by functions on $\Z$ and $\{u_I\}$ as above extends to an exterior algebra $\Omega(J^\infty)$ generated by $\Omega(\Z)$, the $u_I$ and additional generators $\extd u_I$ with relations 
    \begin{align*}
       [ e^a, u_I ]= u_{a I} e^{a}, \qquad
       [\extd u_I,  \phi ] = \sum_a u_{a I}(\del_a\phi) e^{a},\quad [\extd u_I,  u_J ] =\sum_a u_{a I}u_{a J}e^a
    \\
 \{ e^a, \extd u_I\} +\extd u_{a I} \wedge e^{a}=0,\quad
   \{\extd u_I , \extd u_J\}    +\sum_a (u_{a I} \extd u_{a J} + (\extd u_{a I}) u_{a J}) \wedge e^{a}=0.
    \end{align*}
 The relations including $\extd u_I$ can also be characterised in terms of vertical differentials obeying
    \[ [\extd_V u_I, \Phi] = 0,\quad \{e^a,\extd_V u_I \}+\extd_V u_{a I} \wedge e^{a}=0,\quad  \{\extd_V u_I ,\extd_V u_J \} = 0,\]
for all $\Phi \in C^\infty(J^\infty)$ and  horizontal differential $\extd_H=[\theta,\ \}$ with $\theta=\sum_a e^a$. \end{theorem}
\begin{proof} The degree 1 commutation relations were obtained as~\eqref{eurelns},~\eqref{duphi} and~\eqref{duurelns} but one can also check directly that they give a first-order calculus on $C(J^\infty)$. We then apply $\extd$ to the degree 1 relations as explained and necessarily get an exterior algebra (the canonical maximal prolongation~\cite{BegMa} modulo further relations in $\Omega(\Z)$).  We also explained  natural decomposition to $\extd_H+\extd_V$ such that  $\extd_H$ remains inner so that $\extd u_I=\sum_a u_{a I}e^a + \extd_V u_I$. In this case,
\begin{align*}[ \extd_Vu_I, \phi]&=[\extd u_I,\phi]-\sum_a u_{a I}[e^a,\phi]=[\extd u_I,\phi]-\sum_{a} u_{a I}(\del_a\phi)e^a=0,\\
 [\extd_Vu_I, u_J]&=[\extd u_I,u_J]-\sum_a u_{a I}[e^a,u_J]=0,\end{align*}
using the commutation relations of $e^a$ with $\phi\in C(\Z)$ and with $u_J$. Hence $\extd_V u_I$ are central as claimed. Similarly substituting $\extd u_I$ gives the $\{e^a,\extd_Vu_I\}$ relation. We then use these results on expanding out both sides of the $\{\extd u_I,\extd u_J\}$ relation in terms of $\extd_V u_I$ and $e^a$ etc., to eventually obtain the last relation.
 \end{proof}

This is the natural quadratic extension of the degree 1 relations and splitting of $\extd$ that keeps $\extd_H$ inner. We see that $e^a, \extd_V u_I$ together generate a closed algebra, which we denote $\Lambda_{J^\infty}$ (over the field) and that the relations allow everything to be ordered with functions to the left so that this algebra provides a natural basis over $C(J^\infty)$ in each degree. The $\extd_V u_I$ are, moreover, a central basis as well as form as well as anticommute among themselves. We do not exclude the possibility of additional relations in higher degree beyond ones implied by these degree 2 relations, but we do not appear to need them.

\subsection{Euler-Lagrange equations}
\label{sec:ELeq}
We proceed in a similar manner as the classical case, now with $L \Vol$, for $L = L(u,u_a)$ a first-order Lagrangian and $\Vol = e^+\wedge e^-$. For the following calculations it is useful to define the operators
\[
R_a \colon \Omega^k(J^\infty) \to \Omega^k(J^\infty), \qquad
D_a \colon \Omega^k(J^\infty) \to \Omega^k(J^\infty),
\]
via 
\[
e^a \omega = (-1)^{|\omega|}R_a(\omega) e^a,
\qquad 
\extd_H \omega = (-1)^{|\omega|} D_a (\omega) e^a,
\]
as specified by the commutation relations in Section~\ref{secOmegaJ}. These are related as $D_a = R_a - \id$ and have the following properties
\begin{equation}
    \label{eq:propRD}
    R_a(\omega\wedge \eta) = R_a(\omega) \wedge R_a(\eta), \qquad
    D_a(\omega\wedge \eta) = D_a(\omega)\wedge R_a \eta + \omega \wedge D_a(\eta).
\end{equation}
Note that in degree 0, $R_a(u_I)= u_I+ u_{a I}, D_a(u_I)=u_{a I}$. These operations commute with each other and satisfy the following identities
\[
R_a R_{a^{-1}} = \id, \qquad D_a R_{a^{-1}} = R_{a^{-1}}D_a = -D_{a^{-1}}, \qquad D_{a} D_{a^{-1}} = -D_a - D_{a^{-1}},
\]
so that for example $u_{aa^{-1}} =  D_{a^{-1}} D_{a}u = -u_a - u_{a^{-1}}$. Equations~\eqref{duphi}-\eqref{duurelns} can now be combined as 
\[ [\extd u,\Phi]=\sum_a u_{a I}(D_a\Phi)e^a\]
for all $\Phi\in C(J^\infty)$.
We also choose elements  $\Vol_a$ such that $\Vol = e^a \wedge \Vol_a$ (no sum) (in our example, we will let  $\Vol_+ = e^-$, $\Vol_- = -e^+$).

\begin{theorem}
\label{thm:ELandBoundaryTerm}
For $L = L(u,u_a)$, we define  (with sums over $a$),
\[
EL = \left(\frac{\del L}{\del u} + \sum_a D_{a^{-1}}\left(\frac{\del L}{\del u_a}\right)\right) \extd_V u \wedge \Vol, \qquad
\Theta = \sum_a R_{a^{-1}}\left(\frac{\del L}{\del u_a}\right)\extd_V u \wedge \Vol_a.
\]
Then $\extd_V(L \Vol) = EL - \extd_H \Theta$.
\end{theorem}

\begin{proof}
This is a matter of computation. We start with
\[
\extd_V(L \Vol)
    = \left(\frac{\del L}{\del u} \extd_V u + \sum_a \frac{\del L}{\del u_a} \extd_V u_a\right)\wedge \Vol
    = \left(\frac{\del L}{\del u} \extd_V u + \sum_a \frac{\del L}{\del u_a} D_a \extd_V u\right)\wedge \Vol
\]
where in the last equality we have used $\extd_H \extd_V u = \{\theta,\extd_V u\} = -\extd_V u_{a} \wedge e^a$ and thus $D_a \extd_V u = \extd_V u_{a}$. Using now the Leibniz rule for $D_a$ \eqref{eq:propRD}
\begin{align*}
    &= \left(\frac{\del L}{\del u} \extd_V u + \sum_a D_a\left(\frac{\del L}{\del u_a} \extd_V u\right) - \sum_a D_a\left(\frac{\del L}{\del u_a}\right) R_a \extd_V u\right)\wedge \Vol\\
    &= \left(\frac{\del L}{\del u} \extd_V u + \sum_a D_a\left(\frac{\del L}{\del u_a} \extd_V u\right) + \sum_a R_a\left( D_{a^{-1}}\left(\frac{\del L}{\del u_a}\right)\extd_V u\right)\right)\wedge \Vol
\end{align*}
where in the second line we have used $-D_a = R_aD_{a^{-1}}$ and the property \eqref{eq:propRD} for $R_a$. Writing now $R_a = D_a +\id$ results in
\begin{align*}
    &= \left(\frac{\del L}{\del u} \extd_V u + \sum_a D_a\left(\frac{\del L}{\del u_a} \extd_V u\right) + \sum_a D_a\left( D_{a^{-1}}\left(\frac{\del L}{\del u_a}\right)\extd_V u\right) + \sum_a D_{a^{-1}}\left(\frac{\del L}{\del u_a}\right)\extd_V u \right)\wedge \Vol.
\end{align*}
Since $D_{a^{-1}} = R_{a^{-1}} - \id$ the third term will cancellation of the second term in the expression
\begin{align*}
    &= \left(\frac{\del L}{\del u} \extd_V u +\sum_a D_a\left( R_{a^{-1}}\left(\frac{\del L}{\del u_a}\right)\extd_V u\right) + \sum_aD_{a^{-1}}\left(\frac{\del L}{\del u_a}\right)\extd_V u \right)\wedge \Vol.
\end{align*}
The last step is to use $\extd_H(f\Vol_a) = \sum_b D_b(f)e^b \wedge \Vol_a = D_a(f) \Vol$ due to the definition of $\Vol_a$. Applying this to the second term of the above expression results in
\begin{align*}
    \extd_V(L \Vol )= \left(\frac{\del L}{\del u} + \sum_a D_{a^{-1}}\left(\frac{\del L}{\del u_a}\right)\right) \extd_V u \wedge \Vol - \extd_H \left( \sum_a R_{a^{-1}}\left(\frac{\del L}{\del u_a}\right)\extd_V u \wedge \Vol_a \right). 
\end{align*}
\end{proof}

In the continuum limit, where we approximate $\Z$ to $\R$ with coordinate $t$,  we expect $D_+$ and $D_-$ to correspond to the positive and negative total derivatives in the $t$ direction $\pm D_1$, and hence we recover the classical EL form in equation \eqref{eq:ClELBound}. Similarly we expect the shift operation $R_{a}$ to correspond to the identity in the limit, recovering the boundary form in \eqref{eq:ClELBound}. Note that the relation $D_a=R_a-\id$ now needs to be scaled so that the finite difference becomes a usual derivative, which then also enters $R_a u$.

As the bundle $\Z \times \R \to \Z$ models a point particle moving in $\R$ evolving in discrete time, we want to reproduce Example~\ref{ex:ClMech} in this setting, starting with a free particle $q=\phi:\Z\to \R$ with $V(q) = 0$. Let $(e^a,e^b) = g^{ab} = \delta^{a,b^{-1}}$ be the inverse Euclidean metric on $\Z$, and define the integral on $\Z$ as $ \int_\Z f \Vol = \sum_i f_i$. Then in analogy to Example~\ref{ex:ClMech} we set the action to be
\begin{align*}
    S[q] &= - \frac{m}{4} \int_\Z (\extd q,\extd q) \Vol 
    = -\frac{m}{4} \int_\Z \sum_{a,b} (\del_a q) (R_a \del_b q) (e^a,e^b) \Vol\\
    &= \frac{m}{4} \int_\Z \left((\del_+ q)^2 + (\del_- q)^2\right) \Vol,
\end{align*}
where we have used $R_a\del_{a^{-1}} = - \del_a$. Note the prefactor $- \frac{m}{4}$ instead of ${m\over 2}$. The minus sign comes from the overall minus sign that appears in $(\extd q,\extd q)$, and the extra $1\over 2$ is introduced to ensure that the continuum limit is correct, since in the limit $\Z \to \R$ we expect $\del_\pm q$ to correspond to $\pm \del_t \phi$, and therefore $(\del_+ q)^2 + (\del_- q)^2$ to $2 (\del_t q)^2$. The Lagrangian is then $L = {m\over 4} \sum_a u^2_a$, with EL and boundary forms given by
\[
    EL = {m\over 2} \sum_a u_{aa^{-1}} \extd_V u \wedge \Vol,
    \qquad
    \Theta = - {m\over 2} \sum_a u_{a^{-1}} \extd_V u \wedge \Vol_a.
\]
In the continuum limit, we expect $\sum_a u_{aa^{-1}}$ and $u_{a^{-1}}$ to correspond to $-2 u_{tt}$ and $-u_t$ respectively, thus recovering the results of Example~\ref{ex:ClMech}. Using now $u_{a a^{-1}} = -u_a - u_{a^{-1}}$, we find the EL equation, or Newton's equation, in discrete time to be
\[
    {m\over 2} \sum_a \del_a \del_{a^{-1}} q= - m \sum_a \del_a q = 0.
\]
We can also add a potential term $V(q)$ to the action as
\[
S[q] =  \int_\Z \left(-\frac{m}{4}(\extd q,\extd q) - V(q)\right) \Vol
\]
so that the Lagrangian is now $L = \frac{m}{4} \sum_a u^2_a - V(u)$. The corresponding EL form and EL equations are now
\[
    EL = \left({m\over 2} \sum_a u_{aa^{-1}} - {\del V(u)\over \del u}\right) \extd_V u \wedge \Vol,
    \qquad 
    m \sum_a \del_a q = - {\del V(q)\over \del q}
\]
with the boundary term is left unchanged.

\subsection{Noether current and conserved energy for classical mechanics $\Z \times \R \to \Z$}
\label{sec:Noether}

Rather than a  general Noether's theorem we will look in the $\Z$ case at the obvious `time translation' symmetry. We continue with the application to (non-relativistic) classical mechanics with discrete time where we write $q=\phi: \Z\to \R$ for the field (the classical field theory interpretation will be covered in Section~\ref{secZm}).

From the general discussion in Section~\ref{sec:SymNoeCl} we take the key requirement of an interior product along a suitable vector field encoding the translation symmetry and we specify this directly as a map
$\iota_\eps \colon \Omega^{\bullet}(J^\infty)\to \Omega^{\bullet-1}(J^\infty)$ by defining it on $\Lambda_{J^\infty}$ and extending as a left module map to $\Omega(J^\infty)=C(J^\infty)\Lambda_{J^\infty}$. Following the treatment for $M$ in Example~\ref{ex:ClStressEnergy}, we let 
\[
    \iota_\eps(e^a) = \epsilon^a,
    \quad 
    \iota_\eps(\extd_V u_I) = -\sum_a \epsilon^a u_{a I}
\]
on the generators, for some parameters $\epsilon^a$. Moreover, we split this as   $\iota_\eps = \iota_H + \iota_V$ where 
\[\iota_H e^a = \eps^a,\qquad \iota_V\extd_Vu_I=-\sum_a\eps^a u_{a I}\]
and $\iota_H\extd_V u_I= \iota_V e^a = 0$. We then extend $\iota_V$ to the $\Lambda_{J^\infty}$ as a graded-derivation. We similarly extend $\iota_H$ to the Grassmann subalgebra $\Lambda_\Z$ generated by  $\{e^a\}$ as a graded derivation. We do not need $\iota_H$ beyond this but one can, for example, extend it as commuting with right-multiplication by $\extd_Vu_I$.

\begin{lemma}\label{lem:IntProd} $\iota_H,\iota_V$ as specified are well-defined, and graded-derivations on $\Lambda_\Z$, $\Lambda_{J^\infty}$ respectively. \end{lemma}
\begin{proof} We need to verify that on this subalgebra $\iota_\eps$ is well-defined for the quadratic relations between these. It is easy to see that $\iota_\eps(\{e^a,e^b\})=0$ and $\iota_\eps(\{\extd_Vu_I,\extd_Vu_J\})=0$ just because $\eps^a$ are numbers and $\extd_Vu_I$ are central. This applies to both $\iota_H,\iota_V$ parts. For the cross-relation, we check using the graded derivation property that 
\begin{align*}\iota_V(\{e^a,\extd_V u_I\})&=e^a \sum_b \eps^b u_{b I}- \sum_b \eps^b u_{b I}e^a =\sum_b \eps^b[e^a,u_{b I}]\\&=\sum_b \eps^b u_{ab I} e^a= \sum_b \eps^b u_{b a I} e^a=-\iota_V(\extd_V u_{a I}e^a),\end{align*}
so $\iota_V$ extends to a graded derivation as claimed. For $\iota_H$, we cannot impose that it is a full derivation as this would imply $\iota_H(\omega\extd_Vu_I)=\iota_H(\omega)\extd_V u_I$ and $\iota_H((\extd_Vu_I)\omega)=-(\extd_V u_I)\iota_H(\omega)$ and hence $\iota_H(\{e^a, \extd_V u_I\})=0$, whereas the value on the other side would be $\eps^a \extd_V u_{a I}$. But we can impose just the first of these (for example) to fully specify it, using the commutation rules to put all $\extd_V u_I$ factors to the right and take them out. 
\end{proof}

In particular, since $\iota_H$ is a graded derivation on $\Lambda_\Z$, we have
\begin{equation}\label{iVol} i_H({\rm Vol})=i_H(e^+\wedge e^-)=\eps^+ e^-- e^+\eps^-=\sum_a \eps^a {\rm Vol}_a.\end{equation}

Focusing now on a free classical particle with the Lagrangian $L = {m\over 4} \sum_a u^2_a$ we define the  `naive Noether current' as 
\[
    j_0 = \sigma -\iota_H(L\Vol) - \iota_V \Theta = - \sum_{a,b} \epsilon^b\left({m\over 2}u_{a^{-1}}u_b + \delta_b^a L \right)\Vol_a
\]
where we copied the classical case by setting $\sigma = 0$, used the left module map property of $\iota_\epsilon$ and~\eqref{iVol}. It turns out that this current is not conserved, but close enough that it can be corrected.

\begin{proposition}
\label{prop:ConservedCurrentClMech}
    The current
    \[
        j = - \sum_{a,b} \epsilon^b \left({m\over 2}\left(u_{a^{-1}}+\frac{1}{2}u_{a^{-1}b}\right)u_b + \delta^{a}_b L \right) \Vol_a
    \]
    is conserved when the EL equations hold for the Lagrangian $L(u,u_a) = \frac{m}{4} \sum_a u^2_a$. 
\end{proposition}
\begin{proof}
    Start with the naive current $j_0$, we have
    \begin{align*}
        \extd_H j_0 = - \sum_{a,b} \epsilon^b \extd_H \left(\left( {m\over 2} u_{a^{-1}}u_b + \delta^{a}_b L \right) \Vol_a\right)
        = - \sum_{a,b} \epsilon^b \left( {m\over 2} D_a (u_{a^{-1}}u_b) + D_b L \right) \Vol.
    \end{align*}
    Using the Leibniz rule for $D_a$ and noting $D_b(u^2_a) = u^2_{ab} + 2u_a u_{ab}$, we find
    \begin{align*}
        \extd_H j_0 &= - {m\over 4}\sum_{a,b} \epsilon^b \left(2 u_{aa^{-1}} (u_b + u_{ab}) + 2 u_{a^{-1}}u_{ab} + u^2_{ab} + 2 u_a u_{ab} \right) \Vol.
    \end{align*}
    Using $u_{a^{-1}} = -R_{a^{-1}} u_a = -(u_a + u_{a a^{-1}})$ this simplifies to
    \begin{align*}
        \extd_H j_0 = - {m\over 4} \sum_{a,b} \epsilon^b \left( 2 u_{a a^{-1}} u_b + u^2_{a b} \right) \Vol
    \end{align*}
    where we recognise the first term as the EL part. For the other term, note now that
    \[
        -D_a(u_{a^{-1}b}u_b) = -u_{aa^{-1}b} (u_b + u_{ab}) - u_{a^{-1}b}u_{ab} = -u_{aa^{-1}b} u_b +  u^2_{ab} 
    \]
   where we used $u_{aa^{-1}b} = - u_{ab} - u_{a^{-1}b}$. We now see that
    \[
        \extd_H j_0 = - {m \over 4} \sum_{a,b} \epsilon^b \left(2u_{aa^{-1}} u_b + D_b (u_{aa^{-1}}) u_b  -D_a(u_{a^{-1}b}u_b) \right) \Vol.
    \]
 The first and second term vanish when the EL equations are fulfilled, and the last can be written as ${m \over 4}\sum_{a, b} \extd_H(u_{a^{-1} b} u_b) \wedge \Vol_a)$, so that the current $$j = j_0 - \sum_{a,b} \frac{m}{4} \epsilon^b(u_{a^{-1}b} u_b) \Vol_a$$ is conserved when the EL equations hold. 
\end{proof}

We see that an unexpected extra term needs to be added to obtain a conserved current, namely $-{m \over 4}\epsilon^b u_{a^{-1}b} u_b \Vol_a$. The noncommutative-geometric origin of this term is not clear but may relate to our assumption that $\sigma=0$.

We now aim to build the stress-energy tensor $T=\sum_{a,b} T_{a b} e^a \tens_A e^b \in \Omega^1 \tens_A \Omega^1$ on the base $A=C(\Z)$. The divergence of any (0,2)-tensor $T$  is defined geometrically as
\[
   ((\,,\,) \tens \id)\nabla (T)
   =\sum_{a,b} (\extd (T_{a b}) ,e^a) e^b = \sum_{a,b}\del_{a^{-1}} (T_{a b} ) e^b
\]
where $(\ ,\ )$ is the metric inner product. In our case $\nabla$ obeys $\nabla e^a=0$, the Leibniz rule and extends to tensor products in an obvious way (it is a bimodule connection in the sense used in~\cite{BegMa} with the generalised braiding given by flip on the basis). Hence divergence-free just amounts to $\sum_a \del_{a^{-1}} (T_{a b}) =0$ for all $b$.

With this in mind, working `upstairs' on $C(J^\infty)$, we  therefore look for $T_{ab}$ such that

\begin{equation}
    \label{eq:SETfromCurrent}
j = \sum_{a,b} \epsilon^b T^a{}_{b} \Vol_{a}
\end{equation}
and then recover $T_{ab}$ by lowering the indices using the metric
\[
    T_{ab} = \sum_c g_{ac}T^c{}_b = \sum_c \delta_{a,c^{-1}}T^c{}_b = T^{a^{-1}}{}_b.
\]
Then $\extd_H j = 0$ (on-shell) implies $$\sum_a D_{a} T^{a}{}_b = \sum_a D_{a} T_{a^{-1}b}= \sum_a D_{a^{-1}} T_{ab} = 0$$ (on-shell) for all $b$. In our case, we can then read off from Proposition~\ref{prop:ConservedCurrentClMech}
\[
    T_{a b} = - {m\over 2} \left(u_{a} u_b+\frac{1}{2}u_{a b} u_b\right) - \delta^{a^{-1}}_b L.
\]
Note that this stress-energy tensor is not symmetric as $T_{a b} - T_{b a} = -\frac{m}{4}u_{a b} (u_b - u_a)$.

In terms of our field $q:\Z\to \R$ as in Example~\ref{ex:ClMech}, we thus find
\begin{align*}
    &T_{++}[q] = - {m\over 2}\left((\del_{+}q)^2 + \frac{1}{2} (\del^2_{+}q) \del_{+}q\right),&
    &T_{+-}[q] = -{m\over 4} \del_+ q\left(\del_- q + \del_+ q\right), \\
    &T_{-+}[q] = -{m\over 4} \del_- q\left(\del_+ q + \del_- q\right),&
    &T_{--}[q] = - {m\over 2}\left((\del_{-}q)^2 + \frac{1}{2} (\del^2_{-}q) \del_{-}q\right),
\end{align*}
where we used $\del_+ \del_- = -\del_+ -\del_-$.

\begin{corollary}
    \label{cor:indeqDivT}
    The divergence free condition for the stress-energy tensor $\sum_a D_{a^{-1}} T_{a b} = 0$ for $b=+,-$ only gives one independent condition.
\end{corollary}

\begin{proof}
To see this we want to relate the divergence free condition for $b$ and $b^{-1}$. Comparing the expression for $\extd_H j$ from Proposition~\ref{prop:ConservedCurrentClMech} with the decomposition in equation \eqref{eq:SETfromCurrent} we can relate the divergence of the stress-energy tensor to an expression which depends on the EL equations. From there we compute
\begin{align*}
    \sum_{a} D_{a^{-1}} T_{a b} 
    &= -{m\over 4} \sum_a \left( 2u_{a a^{-1}}+  D_b\left(u_{a a^{-1}} \right)\right)  u_b\\
    &= -{m\over 4} \sum_a  \left( 2u_{a a^{-1}} u_{b} + R_b\left(D_{b^{-1}}\left(u_{a a^{-1}} \right) u_{b^{-1}}\right)\right).
\end{align*}
In the case of $b^{-1}$ we have 
\[
\sum_{a} D_{a^{-1}} T_{a b^{-1}} 
    = -{m\over 4} \sum_a \left( 2u_{a a^{-1}}+  D_{b^{-1}}\left(u_{a a^{-1}} \right)\right)u_{b^{-1}}.
\]
We therefore see that the last term in the changed expression for $\sum_{a} D_{a^{-1}} T_{a b}$ can be written in terms of $\sum_{a} D_{a^{-1}} T_{a b^{-1}}$, leading to
\[
\sum_{a} D_{a^{-1}} T_{a b} 
    = -{m\over 4} \sum_a 2(R_b + \id)(u_{a a^{-1}}) u_{b} + \sum_{a} D_{a^{-1}} T_{a b^{-1}}.
\]
The first term is just the EL equations, and therefore we see that the divergence free condition for $b$ and $b^{-1}$ are the same on-shell.
\end{proof}

Looking at the content for the $b=+$ case, we have
\[
    0 = D_- T_{++} + D_+ T_{-+} = D_+(T_{-+} -R_- T_{++}),
\]
where we have used $D_a = - D_{a^{-1}} R_a$, and directly see that the term in brackets is constant along the discrete time axis $\Z$. Expanding it out gives
\begin{align*}
    T_{-+} -R_- T_{++} = -{m\over 4} \left(u^2_- + u_+ u_- - R_-(2 u^2_+ + u_{++}u_+) \right) = -{m\over 2} u_+ u_-
\end{align*}
where we have used $R_-(u^2_+) = u^2_-$ and $R_-(u_{++}u_+) = u_{+- }u_-$. We see that this expression has a quadratic term in the derivatives. In terms of the field $q$, we identify this as the (kinetic) energy
\[
    E[q] = -{m\over 2} (\del_+ q) (\del_- q).
\]
The minus sign is important in order to recover the kinetic energy in the classical limit from Example~\ref{ex:ClMech}. In the continuum limit we expect as discussed $\del_\pm q$ to correspond to $\pm\del_tq$, making that the sign of the above kinetic energy is positive as expected. 

Since this is a new formalism, we check the on-shell conservation claim explicitly. Using the modified Leibniz rule for finite differences,
\begin{align*}
\del_+((\del_+ q) (\del_-q)) &= (\del^2_+ q) R_+(\del_-q) + (\del_+ q) (\del_+ \del_-q)\\
&= R_+((\del_+ \del_-q)) \del_+ q + (\del_+ q) (\del_+ \del_-q)
\end{align*}
where we used $R_+\del_-=-\del_+$ for the second equality. We see that the EL term appears in both summands, making it manifestly zero if $q$ obeys the equations of motion. 

\begin{corollary}
The current from Proposition~\ref{prop:ConservedCurrentClMech} is conserved when a quadratic potential $V(q) = \mu  q^2$ is included in the Lagrangian $L(u,u_a) = {m\over 4} \sum_a u^2_a - V(u)$.
\end{corollary}
\begin{proof}
Following the calculations from Proposition~\ref{prop:ConservedCurrentClMech}, we now find
\[
    \extd_H j = - {m \over 4} \sum_{a,b} \epsilon^b \left(2u_{aa^{-1}} u_b + D_b (u_{aa^{-1}}) u_b \right) + \sum_{b} \epsilon^b (D_b V(u)) \Vol
\]
for $V(u) = \mu u^2$. Using the identity $D_b V(u) = {\del V(u)\over \del u} + {1\over 2}D_b{\del V(u)\over \del u} $ for the quadratic potential the above reduces to 
\[
    \extd_H j = - \sum_{a,b} \epsilon^b \left(\left({m \over 2}u_{aa^{-1}} - {\del V(u)\over \del u}\right) +\frac{1}{2} D_b \left({m \over 2} u_{aa^{-1}}  - {\del V(u)\over \del u}\right) u_b \right) \Vol
\]
which vanishes when the EL equations are satisfied.
\end{proof}

Repeating the analysis with the stress-energy tensor, the charge now corresponds to the total energy
\begin{equation}\label{Eq}
E[q] = -{m\over 2} (\del_+q)(\del_-q) + V(q)\end{equation}
conserved on-shell for the equations of motion
\begin{equation}\label{ELq} \Delta_\Z q =- {1\over m}{\del V(q)\over \del q}\end{equation} 
when the potential is quadratic. It is not clear how to extend these results to higher order potentials, but we are able to construct a discrete time harmonic oscillator and we now look at this in more detail. 

First of all, the left hand side of~\eqref{ELq} is the standard lattice double-differential and writing $V(q)={1\over 2}{m\omega^2} q^2$ for frequency $\omega$, the equation becomes $\Delta_\Z q=-\omega^2 q$. The solution to this on a lattice is well-known to be 
\[ q(i)= q_1 \cos(\omega' i)+ q_2 \sin(\omega' i);\quad \omega'=\arccos\left(1-{\omega^2\over 2}\right)= \omega+{\omega^3\over 24}+O(\omega^5)\]
provided $\omega<2$. A short calculation gives energy
\begin{equation}\label{Eq1q2} E[q]=\frac{m}{2} \sin^2(\omega') (q_1^2+q_2^2)= {m\over 8}\omega^2(4-\omega^2)(q_1^2 + q_2^2)\end{equation}
which we see is $\ge 0$, and zero only for the zero solution. The mix of sine and cosine modes depends on the parameters  and here is instructive to take them to be the values of $q(0)$ and $E[q]$. One can analyse this analytically or just go back to the equations of motion, solved by recursion with
\[ q(1)={q(0)\over 2}\left(2-\omega^2+\sqrt{8 {E[q]\over  m q(0)^2}-4\omega^2+\omega^4}\right)\]
From this we see that there are real solutions if and only if  
\begin{equation}
    \label{Ec} 
    E[q]\ge E_c:={m\over 8}\omega^2(4-\omega^2)q(0)^2 = \left(1-\frac{\omega^2}{4}\right)V(q(0)).
\end{equation}
In classical mechanics, the total energy cannot be less than the potential energy and here, we find a deformation of this fact (here at $i=0$) as recovered for small $\omega$. For example, we can take $q(0)=0$ then $E_c=0$ and we see only sine solutions increasing in amplitude with increasing $E>0$ according to $E={1\over 8}m \omega^2 (4-\omega^2)q_2^2$ from~\eqref{Eq1q2}. This modifies the relation $E={1\over 2}m \omega^2 q_2^2$ in classical mechanics, parallel to the familiar fact in lattice field theory that the dispersion relation gets modified when $\omega$ is not small. Also parallel to the field theory case, if we let $\omega^2$ be negative then there are (real) unbounded exponential solutions. As we increase $\omega$, the amplitude for a given energy decreases until a minimum at $\omega=\sqrt{2}$, while the period $2\pi/\omega'$ shrinks towards period 2, i.e alternating and blowing up in amplitude as  $\omega\to 2$ from below. The case $\omega>2$ is not physical (if the discretisation is a model of Planck scale effects then these are transplankian modes) and here solutions blow up exponentially while alternating in $i$ (the `maximum frequency'). These effects are illustrated in Figure~\ref{fig1}.  Note that   $E<E_c$ leads to complex oscillatory solutions but these are not part of our theory, and indeed the conserved  energy for the complex case (which would be relevant to the scalar field theory point of view) is different as it involves conjugation of one of the fields. 

\begin{figure}[h]
\[ \includegraphics[scale=0.65]{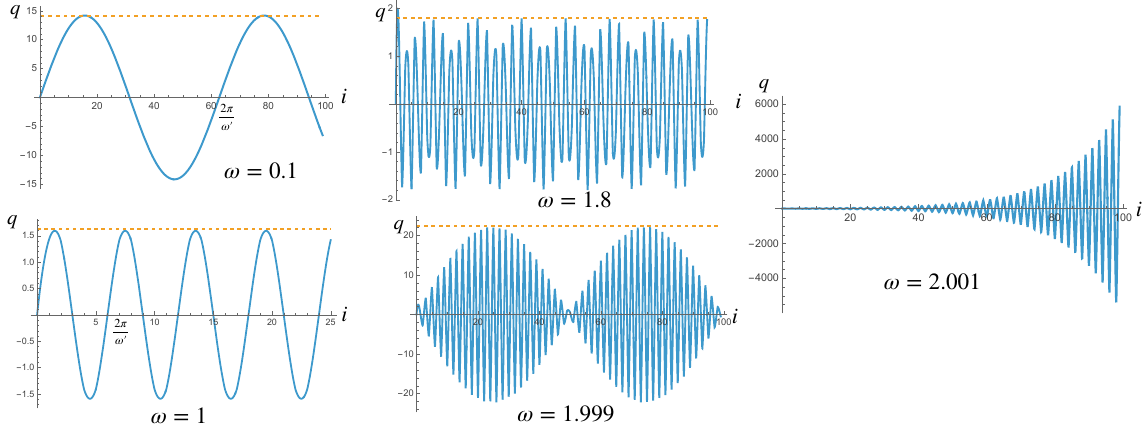}\]
\caption{\label{fig1} Discrete time harmonic oscillator for $m=1$ and energy $E=1$ with increasing $\omega$. For the sine modes with $\omega<2$, the dashed line is the amplitude $q_2$ such that $E={1\over 8}m \omega^2(4-\omega^2)q_2^2$, with minimum amplitude for $\omega=\sqrt{2}$. The period  $2\pi/\omega'$ is significantly smaller than $2\pi/\omega$ as  $\omega$ increases. As $\omega\to 2$ from below, the amplitude blows up and the period tends towards the minimum  alternating case. We also see a beat frequency emerging. The $\omega>2$ case has alternating exponential growth.  Plots have been smoothly interpolated for visualization.}
\end{figure}

Finally, the stress energy tensor, e.g. for the sine solutions $q(i)=q_2 \sin(i \omega')$ is 
\begin{align*}
    T_{++}[q] &= -\frac{m}{4}  q_{2}^2 (\sin (i \omega')-\sin ((i+1) \omega')) (\sin (i \omega')-\sin ((i+2) \omega')),\\
    T_{+-}[q] &= -\frac{m}{8}  q_{2}^2 (-2 \sin (\omega') \sin (2 i \omega')+\cos (2 i \omega')-\cos (2 (i+1) \omega')+\cos (2 \omega')-1),\\
    T_{-+}[q] &= -\frac{m}{8}  q_{2}^2 \left(2 \sin (\omega') \sin (2 i \omega')-2 \cos ^2((i-1) \omega')+\cos (2 i \omega')+\cos (2 \omega')\right),\\
    T_{--}[q] &= \frac{m}{4} q_{2}^2 (\sin ((i-2) \omega')-\sin (i \omega')) (\sin (i \omega')+\sin (\omega'-i \omega')),
\end{align*}
and is not itself constant in $i$ (but is divergence free by our construction). 

\section{Conserved stress-energy tensor on Abelian groups and $\Z^m$ lattices.}
\label{secZm}

We next observe that it is immediate to replace $\Z$ by any Abelian group $G$ as base in order to consider the fibre bundle $G \times \R \to G$, as a model for scalar field theory on $G$. Translation-invariant calculi are of the form $\Omega(G)=C(G)\Lambda_G$ where the invariant forms $\Lambda_G$ again provide a basis over the algebra and are generated by a set of invariant 1-forms $\{e^a\}$. Here $a\in \CC$ can be viewed as labelling group elements in a subset $\CC\subseteq G\setminus\{e\}$, where $e$ is group identity and for a connected calculus (which we assume) we need that $\CC$ generates the group, and for existence of a metric we also need that $\CC$ is closed under group inversion. The set $\CC$ generates a Cayley graph on $G$, where we take the set of vertices to be $G$ with arrows given by right multiplication as $x\to xa$. We define $R_a(f)(x)=f(xa)$ if we denote the group multiplicatively, $\del_a=R_a-\id$ and $\extd f= \sum_a(\del_a f)e^a$ much as before. We also have $\extd=[\theta,\  \}$ for $\theta=\sum_a e^a$. In the Abelian case $\Lambda_G$ is again the Grassmann algebra on the $\{e^a\}$ and $\extd e^a=0$ (the nonAbelian case is similar but a more complicated algebra). In the case of $\Z$, the  smallest choice is  $\CC=\{\pm 1\}$ which corresponds to $e^\pm$ in Section~\ref{secZ}.

Proceeding as before, we still have a torsion free flat connection defined by $\nabla e^a=0$ and since, for $G$ Abelian, the $\del_a$ commute with each other, the jet bundle construction in~\cite{MaSim1} still gives the space of sections $\CJ^\infty$ built on symmetric tensor powers of $e^a$ with jet prolongation map $j_\infty(f)=\sum_{I}\del_If e^I$ where now the multiindex $I=(a_1\cdots\cdots a_m)$ enumerates over symmetrized indices $a_i\in \CC$ and $\del_I$ is the iterated derivative (and we include $I=\emptyset$ as no derivative in the sum). The index $a I$ denotes $(a a_1\cdots a_m)$ (it is not necessary  but we can also fix everything explicitly by choosing an ordering on $\CC$ and taking representatives that are lexicographically ordered).

We then coordinatise the fibre of $J^\infty$ by $u_I$ corresponding to the coefficient of $e^I$ and have $C(J^\infty)=C(G)[u_I]$ as before. The evaluation map is $e_\infty: G\times F\to J^\infty$ where $F$ is the space of fields (for a rank 1 real bundle it is just another copy of $C(G)$ up to completions) and pulling back gives $e_\infty^*: \Omega(J^\infty)\to \Omega(G)\underline\tens \Omega(F)$. Assuming (for the purposes of deriving the formulae) that this is an inclusion gives the same results as Theorem~\ref{thm:Omega2rel}. The only difference is that $\Z$ is replaced by $G$ and now $a\in \CC$ as discussed rather than $\pm$ as before.

Next, for the Euler-Lagrange equations, we still have a unique (up to scale) top form ${\rm Vol}=e^1\wedge \cdots \wedge e^{|\CC|}$ and we use this to define $\Vol_a$ as a product of the 1-forms without $e^a$ such that $\Vol = e^a \wedge \Vol_a$ (no sum). Then the computation of $\extd_V (L \Vol)$ and hence of the EL form and boundary form goes as in Theorem~\ref{thm:ELandBoundaryTerm}, yielding the same results, now with $a \in \CC$. Similarly, Section~\ref{sec:Noether} can be reproduced up to the introduction of a Lagrangian and correspondingly the form of the stress-energy tensor $T_{a b}$, but the computation of the conserved charges needs extra care, and is also a little different when we make the interpretation 1+0 classical field theory. We will give the results in detail for $\Z^m,\Z^{1,m-1}$, but the method works in the same way for any Abelian group of the form $\Z\times G$ or $\Z_N \times G$, replacing $\CC_{m-1}$ as defined below by $\CC_G$.

\subsection{EL equations and translation symmetries for $X=\Z^m, \Z^{1,m-1}$}

We now focus on scalar field theory on $X=\Z^m$ with the Euclidean metric, meaning that $(e^a,e^b) = \delta^{a,b^{-1}}$. For the calculus, we take $\CC$ to be the set of all the positive and negative directions in $\Z^m$ $$\CC = \{\pm v^i = \pm(0,\dots,0,1,0,\dots,0)\in \Z^m\vert i = 1,\dots, m\}$$
where the $\pm1$ is in the $i$-th position. The volume form is then
\[
    \Vol_{\Z^m} = e^{1+} \wedge e^{1-} \wedge \cdots \wedge e^{m+} \wedge e^{m-}, 
\]
where $e^{i\pm}$ corresponds to the direction $\pm v^i$.

Consider now the following action for free scalar field theory on $\Z^m$
\[
    S[\phi] = \frac{1}{2}\int_{\Z^m} \left(-\frac{1}{2}(\extd \phi,\extd \phi) -m^2 \phi^2\right)\Vol_{\Z^m}
    = \frac{1}{2}\int_{\Z^m} \left(\frac{1}{2} \sum_{a\in \CC} (\del_a \phi)^2 -m^2 \phi^2\right)\Vol_{\Z^m}
\] 
where we again introduce the factor $-{1\over 2}$ in front of the kinetic term in order to recover the classical action from Example~\ref{ELR}. Recall that the minus cancels the one coming from $(\extd \phi,\extd \phi)$, and the ${1\over 2}$ is there since every classical direction corresponds to two directions in $\CC$ (namely $e^{i\pm}$). Given the Lagrangian $L = \frac{1}{4} \sum_{a\in \CC} u_a^2 - \frac{1}{2} m^2 u^2$, and we can use Theorem~\ref{thm:ELandBoundaryTerm} to compute the EL and boundary forms
\[
    EL = \left(- m^2 u + \frac{1}{2} \sum_{a\in \CC}u_{aa^{-1}}\right) d_V u \wedge \Vol, 
    \qquad
    \Theta = -\frac{1}{2} \sum_{a\in \CC}u_{a^{-1}}d_V u \wedge \Vol_a.
\]
We see directly that the continuum limit is expected to correspond to the one in Example~\ref{ELR}. Using $u_{aa^{-1}} = -u_a -u_{a^{-1}}$ the EL equations for the field $\phi$ are then
\[
    \left(\sum_{a\in \CC} \del_a + m^2\right)\phi = 0.
\]
Here 
\[\Delta= \sum_{a\in \CC}\del_a= -{1\over 2}(\ ,\ )\nabla\extd \]
is the standard graph or lattice Laplacian with the Euclidean (constant) metric on edges and $\nabla(\extd e_a)=0$, cf.~\cite{BegMa}. We thus recover the Klein-Gordon equation $(\Delta +m^2)\phi=0$ for this case. The $1/2$ relates to the doubling of derivatives.

Translation symmetry can be set up as in Section~\ref{sec:Noether}, by defining the interior product $\iota_\epsilon$ through $\iota_\eps(e^a) = \epsilon^a$, $\iota_\eps(\extd_V u_I) = -\sum_{a\in \CC} \epsilon^a u_{a I}$ and extending it as in Lemma~\ref{lem:IntProd}. The `naive Noether current' is then 
\[
    j_0 = \sigma -\iota_H(L\Vol) - \iota_V \Theta = -\sum_{a,b\in \CC} \epsilon^b \left(\frac{1}{2}u_{a^{-1}}u_b + \delta^{a}_b L \right) \Vol_a
\]
where we copied the classical case from Example~\ref{ex:ClStressEnergy} with $\sigma = 0$, used the left module map property and~\eqref{iVol}. Again this current is not conserved but close enough to be corrected.

\begin{proposition}
\label{prop:ConservedCurrentScalarField}
    The current
    \[
        j = - \sum_{a,b\in \CC} \epsilon^b \left(\frac{1}{2}\left(u_{a^{-1}}+\frac{1}{2}u_{b a^{-1}}\right)u_b + \delta^{a}_b L \right) \Vol_a
    \]
    is conserved when the EL equations hold for the Lagrangian $L = \frac{1}{4}\sum_{a\in \CC} u^2_a -\frac{1}{2} m^2 u^2$. 
\end{proposition}
\begin{proof}
    Start with the naive current $j_0$, we have
    \begin{align*}
        \extd_H j_0 = -\sum_{a,b\in \CC} \epsilon^b \extd_H \left(\left( {1\over 2}u_{a^{-1}}u_b + \delta^{a}_b L \right) \Vol_a\right)
        = - \sum_{b\in \CC} \epsilon^b \left( {1\over 2} \sum_{a\in \CC}D_a (u_{a^{-1}}u_b) + D_b L \right) \Vol.
    \end{align*}
    Using the Leibniz rule for $D_a$ and noting $D_b(u^2) = u^2_b + 2u u_b$, similarly for $D_b(u^2_a)$, we find
    \begin{align*}
        \extd_H j_0 &= - {1\over 2} \sum_{b\in \CC} \epsilon^b \left(\sum_{a\in \CC}\left( u_{a a^{-1}}(u_b + u_{a b}) + u_{a^{-1}} u_{ab} + {1\over 2}(u^2_{ab} + 2u_a u_{ab})\right) - m^2 (u^2_b + 2u u_b) \right) \Vol.
    \end{align*}
    Using $u_{a^{-1}} = -R_{a^{-1}} u_a = -(u_a + u_{a a^{-1}})$ this simplifies to
    \begin{align*}
        \extd_H j_0 = -\sum_{b\in \CC} \epsilon^b \left( \left( {1\over 2}\sum_{a\in \CC}u_{a a^{-1}} - m^2 u\right) u_b +\frac{1}{4}\sum_{a\in \CC} u^2_{a b} - \frac{1}{2}m^2u^2_b\right) \Vol
    \end{align*}
    where we recognise the first term as the EL term. For the other terms, note that
    \begin{align*}
        - \frac{1}{2}m^2u^2_b &= \frac{1}{2} D_b \left(\frac{1}{2}\sum_{a\in \CC}u_{a a^{-1}} - m^2 u\right) u_b - \frac{1}{4} \sum_{a\in \CC} u_{a a^{-1} b} u_b\\
        &= \frac{1}{2} D_b \left(\frac{1}{2}\sum_{a\in \CC}u_{a a^{-1}} - m^2 u\right) u_b + \frac{1}{4} \sum_{a\in \CC}\left(u_{ab} u_b + u_{a^{-1} b} u_b\right)\\
        &= \frac{1}{2} D_b \left(\frac{1}{2}\sum_{a\in \CC}u_{a a^{-1}} - m^2 u\right) u_b - \frac{1}{4} D_a(u_{a^{-1}b}u_b) - \frac{1}{4} \sum_{a\in \CC} u^2_{ab} .
    \end{align*}
    where in the last line we used
    \[
        -D_a(u_{a^{-1}b}u_b) =  - u_{aa^{-1}b}(u_b + u_{a b}) - u_{a^{-1} b}u_{ab}
        = u_{ab}(u_b + u_{ab}) + u_{a^{-1}b}u_b 
    \]
   due to $u_{aa^{-1}b} = -u_{ab} - u_{a^{-1}b}$.
    \begin{align*}
        \extd_H j_0 &= -\sum_{b\in \CC} \epsilon^b \Big(\big( {1\over 2}\sum_{a\in \CC}u_{a a^{-1}} - m^2 u\big) u_b + \frac{1}{2} D_b\big( {1\over 2}\sum_{a\in \CC}u_{a a^{-1}} - m^2 u\big) u_b  \\
        &\qquad\qquad\quad -\frac{1}{4} \sum_{a\in \CC} D_{a}(u_{a^{-1} b} u_b)\Big) \Vol.
    \end{align*}
 The first and second terms vanish when the EL equations are fulfilled, and the last can be written as $-\sum_{a,b\in \CC} \frac{1}{2}\extd_H(u_{a^{-1} b} u_b) \Vol_a)$, so that the current $$j= j_0 + \sum_{a, b\in \CC} \frac{1}{4} (u_{a^{-1} b} u_b) \Vol_a$$ is conserved when the EL equations hold. 
\end{proof}

Constructing the stress-energy tensor works in the same way, by setting $j = \sum_{a,b\in \CC} \epsilon^b T^a{}_{b} \Vol_{a}$ and lowering indices with the metric such that $T_{ab} = T^{a^{-1}}{}_b$. Then $\extd_H j = 0$ (on-shell) implies $\sum_a D_{a^{-1}} T_{a b} = 0$ (on-shell) for all $b$. For free Euclidean scalar field theory we can read off
\[
    T_{a b} = - \left(\frac{1}{2}u_{a}u_b+\frac{1}{4}u_{ab}u_b + \delta^{a^{-1}}_b L \right)
\]
Note that again this stress-energy tensor is not symmetric as $T_{a b} - T_{b a} = \frac{1}{2}u_{a b} (u_a - u_b)$. Similar to Corollary~\ref{cor:indeqDivT} we find

\begin{corollary}
    \label{cor:indeqDivTScalarField}
    The divergence free condition for the stress-energy tensor $\sum_a D_{a^{-1}} T_{a b} = 0$ for all $b$ encodes $|\CC|/2$ independent equations.
\end{corollary}

\begin{proof}
To see this, we relate the divergence free conditions for $b$ and $b^{-1}$. Comparing the expression for $\extd_H j$ from Proposition~\ref{prop:ConservedCurrentScalarField} and $j = \sum_{a,b\in \CC} \epsilon^b T_{a b} \Vol_{a^{-1}}$ we can relate the divergence of the stress-energy tensor to an expression which depends on the EL equations. From there we compute
\begin{align*}
    \sum_{a\in \CC} D_{a^{-1}} T_{a b} 
    &= \left({1\over 2}\sum_{a\in \CC}  u_{a a^{-1}} - m^2 u^2\right) u_b + \frac{1}{2} D_b\left({1\over 2}\sum_{a\in \CC}u_{a a^{-1}} - m^2 u\right) u_b\\
    &= \left({1\over 2}\sum_{a\in \CC}  u_{a a^{-1}} - m^2 u^2\right) u_b - R_{b}\left(\frac{1}{2} D_{b^{-1}}\left({1\over 2}\sum_{a\in \CC}u_{a a^{-1}} - m^2 u\right) u_{b^{-1}}\right).
\end{align*}
Similar to the $\Z$ case in Corollary~\ref{cor:indeqDivT}, the last term in the brackets can be related to the expression for $b^{-1}$, meaning that the two conditions depend on one another.
\end{proof}

To consider conserved charges, we now split $\Z^m = \Z \times \Z^{m-1}$, with corresponding generating sets $\CC_1 = \{t,t^{-1}\}$ and $\CC_{m-1} = \CC \backslash \{\CC_1\}$, where $t,t^{-1}$ are the positive and negative Euclidean time directions $\pm v^1$. The volume form on $\Z^{m-1}$ is then 
\[
    \Vol_{\Z^{m-1}} = e^{2+} \wedge e^{2-} \wedge \cdots \wedge e^{m+} \wedge e^{m-}.
\]
With this split, we can define conserved charges in a similar way to Corollary~\ref{cor:ClConCh}. 
\begin{corollary}\label{corQ}
    With the above split, write $j = j^t \Vol_t + j^{t^{-1}} \Vol_{t^{-1}} + \sum_{b\in \CC_{m-1}} j^b \Vol_b$. The associated {\em conserved charge}
    \[
    Q[\phi] = \int_{\Z^{m-1}} e_\infty^*(j^t-R_{t^{-1}}j^{t^{-1}})(i,\phi) \Vol_{\Z^{m-1}}
    \]
    is conserved in the sense that $\del_t Q = 0$ on-shell. Furthermore, the charge and current densities
    \[
    \rho[\phi] \coloneqq
    e_\infty^*(j^t-R_{t^{-1}}j^{t^{-1}})(\cdot,\phi),
    \quad \quad 
    J^x[\phi] \coloneqq e_\infty^*(j^x-R_{x^{-1}}j^{x^{-1}})(\cdot,\phi)
    \] 
    for $x$ the positive spatial directions in $\Z^{m-1}$, satisfy the continuity equation
    \[
        \del_t \rho + \sum_{x}\del_x J^x = 0
    \]
    on-shell.
\end{corollary}

\begin{proof}
First note that the conservation of $j$ is equivalent to
\[
    D_t j^t + D_{t^{-1}} j^{t^{-1}} + \sum_{b\in \CC_{m-1}} D_{b} j^b
\]
vanishing on-shell. Using $\del_a e_\infty^*(\Phi) = e_\infty^*(D_a\Phi)$, $-D_{a}R_{a^{-1}} = D_{a^{-1}}$ and $\int_{\Z^{m-1}}\Vol_{\Z^{m-1}} \del_b = 0$ for $b\in \CC_{m-1}$, we have
\begin{align*}
    \del_t Q &= \int_{\Z^{m-1}} e_\infty^*(D_t j^t-D_t R_{t^{-1}}j^{t^{-1}})(i,\phi) \Vol_{\Z^{m-1}}\\
    &= - \sum_{b\in \CC_{m-1}} \int_{\Z^{m-1}} \del_b e_\infty^*(j_b)(i,\phi) \Vol_{\Z^{m-1}} = 0
\end{align*}
on-shell. The continuity equation can be derived by rewriting the conservation of $j$ as 
\[
    D_t (j^t - R_{t^{-1}} j^{t^{-1}}) + \sum_{x} D_{x}(j^x - R_{x^{-1}} j^{x^{-1}}).
\]
where $x$ are the spatial directions in $\Z^{m-1}$.
\end{proof}

For translation symmetries, we have $|\CC| / 2$ conserved quantities due to~\ref{cor:indeqDivTScalarField} as $j^a = \sum_{b\in \CC}\epsilon^b T_{a^{-1}b}$, with the energy corresponding to the $t$ direction and momenta corresponding to the $\Z^{m-1}$ spatial directions being
\begin{align*}
    E[\phi] &= \int_{\Z^{m-1}} (T_{t^{-1} t}[\phi] - R_{t^{-1}}T_{t t}[\phi]) \Vol_{\Z^{m-1}}\\
    &= \int_{\Z^{m-1}} {1\over 2}\left( - (\del_t \phi) (\del_{t^{-1}} \phi) - {1\over 2}\sum_{b\in \CC_{m-1}} (\del_b \phi)^2+ m^2 \phi^2 \right) \Vol_{\Z^{m-1}},\\
     P_{b}[\phi] &= \int_{\Z^{m-1}} (T_{t^{-1} b}[\phi] - R_{t^{-1}}T_{t b}[\phi]) \Vol_{\Z^{m-1}}\\
    &= \int_{\Z^{m-1}}{1\over 2}(\del_{t^{-1}}\phi) \left(\del_{b^{-1}} \phi - \del_{b}\phi\right) \Vol_{\Z^{m-1}}.
\end{align*}
where for the momentum $P_b$ we had to use integration by parts and $\int_{\Z^{m-1}}(\del_b f) \Vol_{Z^{m-1}} = 0$ to arrive to $P_b$.

Comparing the expressions to what was found in the continuum case  in Example~\ref{ex:ClStressEnergy}, we see that these are similar. In the classical limit, where we expect $\del_t \phi$, $\del_{t^{-1}} \phi$ to correspond to $\pm\del_0 \phi$, and the spatial derivatives $\del_b \phi$, $\del_{b^{-1}} \phi$ to $\pm\del_i \phi$, we see that these expressions recover exactly the classical forms for $E$ and $P_i$ for the Euclidean metric $g_{00} = g_{ii} = 1$. The extra $1/2$ is there in order to account for the doubling of the derivatives.

In the case of the Minkowski metric, let us define $(e^a,e^b) = \eta^{ab}= \eta^a\delta^{a,b^{-1}}$, with $\eta^t = \eta^{t^{-1}} = 1$ and $\eta^c = \eta^{c^{-1}}= - 1$ for $c\in \CC_{m-1}$. The action is then
\[
    S[\phi] = \frac{1}{2}\int_{\Z^m} \left(-\frac{1}{2}(\extd \phi,\extd \phi) -m^2 \phi^2\right)\Vol_{\Z^m}
    = \frac{1}{2}\int_{\Z^m} \left(\frac{1}{2} \sum_{a\in \CC} \eta^a(\del_a \phi)^2 -m^2 \phi^2\right)\Vol_{\Z^m}
\] 
with the Lagrangian $L = \frac{1}{4} \sum_{a\in \CC} \eta^a u_a^2 - \frac{1}{2} m^2 u^2$, which leads to the following EL and boundary forms
\[
    EL = \left(- m^2 u + \frac{1}{2} \sum_{a\in \CC}\eta^a u_{aa^{-1}}\right) d_V u \wedge \Vol, 
    \qquad
    \Theta = -\frac{1}{2} \sum_{a\in \CC}\eta^a u_{a^{-1}}d_V u \wedge \Vol_a.
\]
The procedure to find the conserved current is similar to before, which now results in
\begin{equation}
    \label{eq:CurrentMin}
    j = - \sum_{a,b\in \CC} \epsilon^b \left(\frac{\eta^a}{2}\left(u_{a^{-1}}+\frac{1}{2}u_{b a^{-1}}\right)u_b + \delta^{a}_b L \right) \Vol_a
\end{equation}
In this case, we want to construct a stress-energy tensor $T$ which satisfies the divergence free condition for the Minkowski metric, where 
\[
   ((\,,\,) \tens \id)\nabla (T)
   =\sum_{a,b} (\extd (T_{a b}) ,e^a) e^b = \sum_{a,b}\eta^{a}\del_{a^{-1}} (T_{a b} ) e^b
\]
since $\eta^a = \eta^{a^{-1}}$. Therefore, working `upstairs' on the jet bundle, we want to find $T$ such that
$\sum_{a\in \CC} \eta^{a} D_{a^{-1}} T_{ab}=0$ on-shell. We again set $$j = \sum_{a,b\in \CC} \epsilon^b T^a{}_{b} \Vol_{a}$$ with $\extd_H j=0$ on-shell corresponding to  $\sum_{a\in \CC} D_{a^{-1}} T^a{}_{b}=0$ on-shell. Now to lower the indices we use the metric $\eta_{ab} = \eta^{ab} = \eta^a \delta^{a,b^{-1}}$ and write
\[
    T_{ab} = \sum_{c\in \CC} \eta_{ac} T^c{}_b = \eta^a T^{a^{-1}}{}_b,
\]
and we see that $T_{ab}$ satisfies the divergence free condition on-shell for the Minkowski metric. From equation \eqref{eq:CurrentMin} we find
\[
    T_{ab}= -\left(\frac{1}{2} u_a u_b +\frac{1}{4} u_{ab} u_b + \frac{\delta^{a^{-1}}_b}{\eta^a} L \right).
\]
The definition of the energy and momenta follow as in the Euclidean case, leading to
\begin{align*}
    E[\phi] &= \int_{\Z^{m-1}} (T_{t^{-1} t}[\phi] - R_{t^{-1}}T_{t t}[\phi]) \Vol_{\Z^{m-1}}\\
    &=\int_{\Z^{m-1}} {1\over 2}\left(-(\del_t \phi) (\del_{t^{-1}} \phi) + {1\over 2}\sum_{b\in \CC_{m-1}} (\del_b \phi)^2+ {m^2}\phi^2 \right) \Vol_{\Z^{m-1}},\\
     P_{b}[\phi] &= \int_{\Z^{m-1}} (T_{t^{-1} b}[\phi] - R_{t^{-1}}T_{t b}[\phi]) \Vol_{\Z^{m-1}}\\
    &= \int_{\Z^{m-1}}{1\over 2}\del_{t^{-1}}\phi \left(\del_{b^{-1}} \phi - \del_{b}\phi\right) \Vol_{\Z^{m-1}}.
\end{align*}
With the same arguments as before, we see that in this case we indeed recover the continuum limit in Example~\ref{ex:ClStressEnergy}.

\subsection{Scalar field theory in the (1+1)-dimensional case}\label{sec1+1}

We focus now on (1+1)-dimensional case of scalar field theory, starting with the Euclidean version where the Klein-Gordon equation $(\Delta + m^2) \phi(t,x)=0$ is solved by the plane wave solutions
\[
    \phi(t,x) = A\cos(\omega t + \kappa x)
\]
satisfying the following dispersion relation for a lattice 
\begin{equation}
    \label{eq:DispRelEuc}
    m^2 = 2(2-\cos(\omega) - \cos(\kappa)).
\end{equation}
as in~\cite{MontMun,RoLat,Smi}. This is normally justified from its appearance in poles in the propagator after Fourier transform to $\omega,\kappa$ as the Fourier conjugate variables, whereas our considerations here will be in line with the continuum version from classical variational calculus in Example~\ref{ex:ClStressEnergy}, for which we have the required limit $m^2 \simeq \omega^2 + \kappa^2$ for $\omega, \kappa \ll 1$ much smaller than the lattice spacing. Note that in Section~\ref{sec:Noether}, we saw that a classical particle $q$ in a quadratic potential $V(q) = \frac{1}{2}m\omega^2 q^2$ only allowed for plane-wave solutions for $\omega < 2$, with a different type of numerical solutions appearing for $\omega > 2$. Here, we similarly find that the mass $m$ of the field for a plane wave (now playing the role of the frequency $\omega$ in the previous example) is also bounded as $m \leq 2\sqrt{2}$ as $\cos(\omega), \cos(\kappa)\in [-1,1]$.

The relation~\eqref{eq:DispRelEuc} can be inverted to give 
\begin{equation}
    \label{eq:FreqEuc}
    \omega(\kappa) = \arccos\left( 2 - \cos(\kappa) -{m^2 \over 2} \right)
\end{equation}
which is defined for all $\kappa \in [-\pi,\pi]$ that satisfy $\cos(\kappa) \in \left[1-{m^2 \over 2},3-{m^2 \over 2}\right] \cap \left[-1,1\right]$. For example for the values of $m=0,2,2\sqrt{2}$ we find $\cos(\kappa) = 1$, $\cos(\kappa) \in [-1,1]$ and $\cos(\kappa) = -1$ respectively. The dispersion relation is illustrated in Figure~\ref{fig:DispEnMomEuc} at essentially these values of $m$, where it is also compared with the continuum counter part shown dashed. We see that the results match for small masses, but are very different as the fields become more massive. In particular, for massless modes we have  $\kappa = 0$, for modes with mass $m=2$ we see that all  $\kappa$ are allowed, and as we approach $m=2\sqrt{2}$ we have that only the modes close to $\kappa = -\pi,\pi$ propagate. We also observe that the dispersion relation~\eqref{eq:DispRelEuc} has an interesting symmetry 
\[ m \mapsto m' = \sqrt{8-m^2},\quad \omega \mapsto \omega' = \pi - \omega,\quad \kappa \mapsto \kappa' = \pi - \kappa,\] 
with the fixed point at $m' = m = 2$. For example,  $m \to 0$, $\omega(\kappa)$ `collapses' to the point $\omega = \kappa = 0$ in Figure~\ref{fig:DispEnMomEuc}, while in the limit $m\to 2\sqrt{2}$, the dispersion relation mirrors this behaviour but towards the point $\omega = \pi$, $\kappa = \pi$. From $\omega(\kappa)$, we find that the phase and group velocities are
\[
    v_{ph} = \frac{\omega(\kappa)}{\kappa},
    \quad \quad
    v_{gr} = \frac{\del \omega(\kappa)}{\del \kappa} = -\frac{\sin (\kappa)}{\sqrt{1-\left(2-\cos (\kappa)-\frac{m^2}{2}\right)^2}},
\]
as illustrated in Figure~\ref{fig:VelEuc}.

\begin{figure}[h!]
    \centering
    \includegraphics[width=1.05\textwidth]{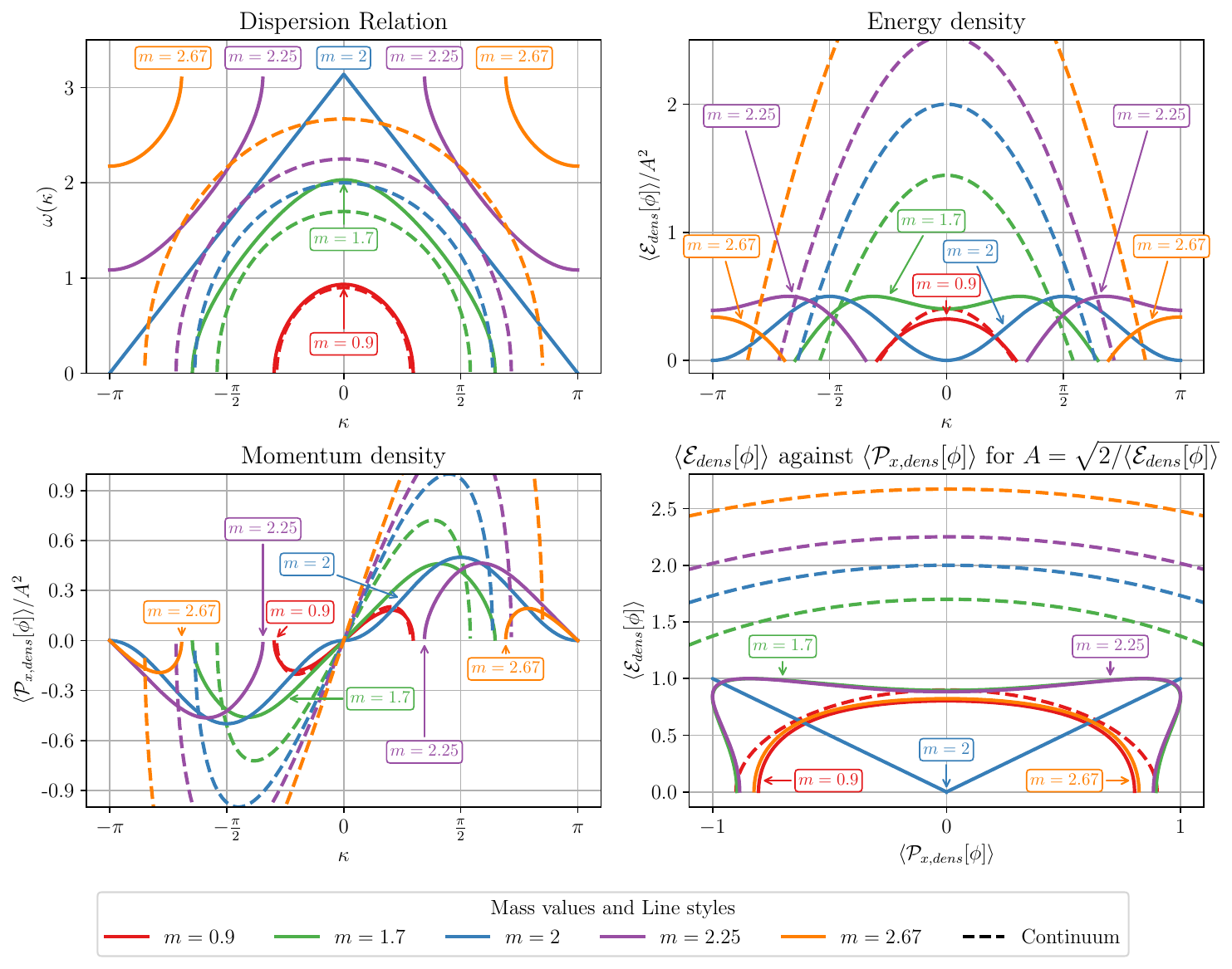}
    \caption{Plots of the dispersion relation, energy and momentum densities at masses $m= 0.9,1.7, 2, 2.25,2.67$ and comparison with the continuum counterparts in the Euclidean case, shown as dashed. The $m>2$ values are the $m'=\sqrt{8-m^2}$ counterparts of the $m<2$ values but with $-0.01$ offset to separate curves in the last plot. For small masses, the lattice and continuum plots almost match. At $m=2$,  waves with all $\kappa \in [-\pi,\pi]$ propagate in the lattice. As we approach $m= 2\sqrt{2}$, waves closer to $\kappa = 0$ stop propagating first, until eventually at $m=2\sqrt{2}$, only waves with $\kappa = -\pi,\pi$ propagate. The lower right plot shows the energy density against the momentum density for $A = \sqrt{2/\<\mathcal{E}_{dens}[\phi] \>}$ as in equation~\eqref{eq:EPrelEuc}.}
    \label{fig:DispEnMomEuc}
\end{figure}

Regarding the energy and momentum densities, we find that they are given by
\begin{align*}
    \mathcal{E}_{dens}[\phi] &= -{1\over 2} (\del_t \phi) (\del_{t^{-1}} \phi) - {1\over 4} (\del_x \phi)^2+ {1\over 2} m^2 \phi^2\\
    &= \frac{A^2}{8} \left(4 \cos ^2(\kappa x+t \omega)-\cos (2 \kappa (x-1)+2 t \omega)-\cos (2 (\kappa x+\kappa+t \omega))-2 \cos (2 \omega)\right),\\
    \mathcal{P}_{x,dens}[\phi] &= {1\over 2}(\del_{t^{-1}}\phi) \left(\del_{x^{-1}}\phi - \del_x \phi \right) = A^2 \sin (\kappa) \sin (\omega) \sin ^2\left(\kappa x+\frac{\kappa}{2}+t \omega-\frac{\omega}{2}\right).
\end{align*}
As we saw in Example~\ref{ex:ClStressEnergy}, these densities do not make it directly clear that the energy $E[\phi]$ and momenta $P_x[\phi]$ are time independent. In the continuum, we took the spatial average over a period $2\pi/\kappa$ and found that it is time independent for both quantities, signifying that the total energy and momentum will also be. In the discrete case, taking the average over a period $2\pi/\kappa$ is not natural as the period is not an integer in general. Instead, we average over the whole $\Z$ line to find
\[
    \<\mathcal{E}_{dens}[\phi] \> = \lim_{N\to \infty} \frac{1}{2N}\sum^{x=N}_{x=-N} \mathcal{E}_{dens}[\phi] = \frac{A^2}{2} \sin ^2(\omega),
\]
\[   
    \<\mathcal{P}_{x,dens}[\phi] \> = \frac{A^2}{2} \sin (\kappa) \sin (\omega).
\]
In both cases, the classical limit where $\omega$ and $\kappa$ are small recovers the expressions from Example~\ref{ex:ClStressEnergy} in the Euclidean case as a useful check on our reasoning. As expected, we found that both the energy and momentum densities are bounded as 
\begin{equation}
    \label{eq:RealScalarEPdens}
    \<\mathcal{E}_{dens}[\phi] \> \in \left[0,\frac{A^2}{2}\right],
    \quad \quad
    \<\mathcal{P}_{x,dens}[\phi] \> \in \left[-\frac{A^2}{2},\frac{A^2}{2}\right],
\end{equation}
due to lattice effects. These densities and their comparison with the classical expressions are also shown in Figure~\ref{fig:DispEnMomEuc}, where we have scaled these quantities by $1/A^2$. The symmetry between $m$ and $m' = \sqrt{8-m^2}$ is also present in the energy and momentum density, where the plots for $m= 0.9,1.7$ centred around $\kappa$ are mirrored by plots for $m'=2.25, 2.67$ now centered around $\kappa = \pi$. We did not take exactly $\sqrt{8-m^2}$ so that the curves would not be exactly on top of each other in the energy-momentum plot.

\begin{figure}[h!]
    \centering
    \includegraphics[width=1\textwidth]{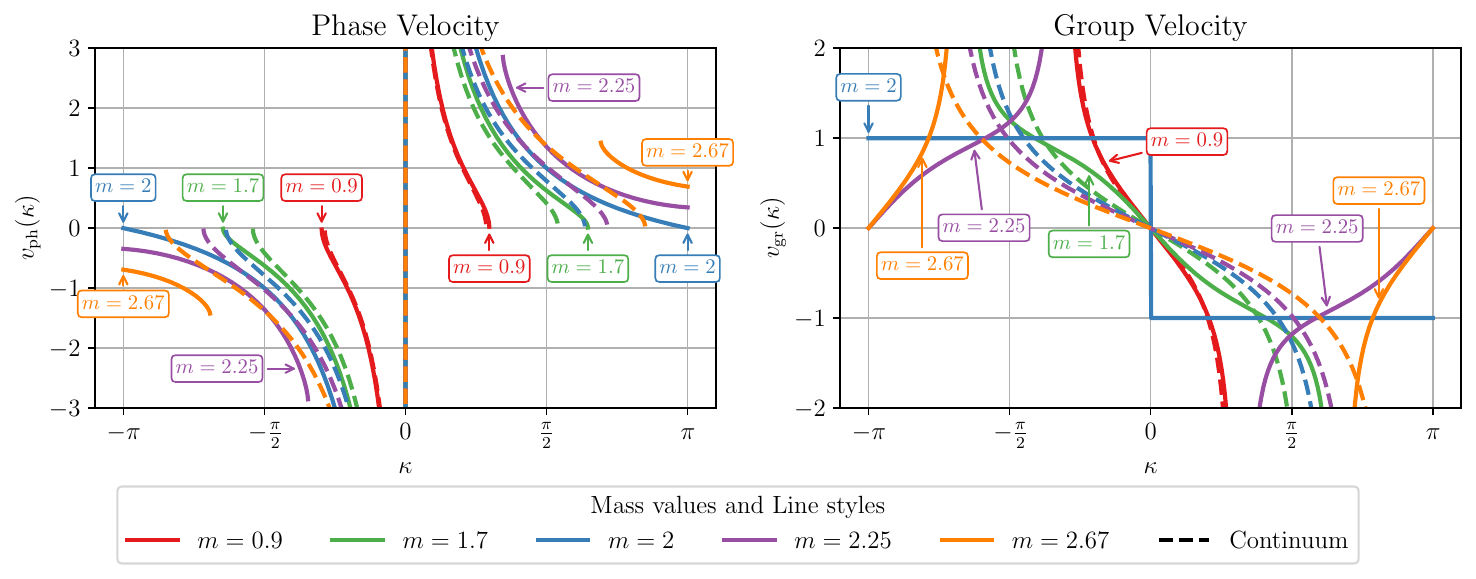}
    \caption{Plots for the phase and group velocity for different masses as in Figure~\ref{fig:DispEnMomEuc}.}
    \label{fig:VelEuc}
\end{figure}

As in Example~\ref{ex:ClStressEnergy}, we find a relation between the energy and momentum densities through the dispersion relation~\eqref{eq:DispRelEuc}. We start by writing $\cos(\omega) = s_\omega \sqrt{1-\sin^2(\omega)}$ and by $\cos(\kappa) = s_\kappa \sqrt{1-\sin^2(\kappa)}$ for $s_\omega, \, s_\kappa\in\{\pm 1\}$. Using the formulas for the energy and momentum density above to write $\sin^2(\omega)$ and $\sin^2(\kappa)$ in terms of $\<\mathcal{E}_{dens}[\phi] \>$, $\<\mathcal{P}_{x,dens}[\phi] \>$, and  reordering the terms results in the energy and momentum densities relation
\begin{align}
    \label{eq:EPrelationEuc}
    \<\mathcal{P}_{x,dens}[\phi]\>^2 &= \frac{A^2}{2} \<\mathcal{E}_{dens}[\phi]\> \left(1 - \left(2 - \frac{m^2}{2} - s_\omega \sqrt{1-\frac{2 \<\mathcal{E}_{dens}[\phi]\>}{A^2}}\right)^2\right)\\
    &\approx \left(- 1 +  \frac{m^2}{2}\right) \<\mathcal{E}_{dens}[\phi]\>^2 + \frac{A^2}{2} \<\mathcal{E}_{dens}[\phi]\> m^2 \left(1 - \frac{m^2}{4}\right),\nonumber
\end{align}
where the approximation was made for $\omega \ll 1$ (i.e. $\<\mathcal{E}_{dens}[\phi]\> \ll 1$) much smaller than the lattice spacing, and thus $s_\omega = 1$ since $\cos(\omega) > 0$ in this case. This clearly deforms the classical relation $\<\mathcal{E}_{dens}[\phi]\>^2 + \<\mathcal{P}_{x,dens}[\phi]\>^2 = (A^2/2)\<\mathcal{E}_{dens}[\phi]\>^2\,m^2$ from Example~\ref{ex:ClStressEnergy} in the Euclidean case. 

Following the strategy in Example~\ref{ex:ClStressEnergy}, we set $A = \sqrt{2/\<\mathcal{E}_{dens}[\phi] \>}$ to have a field $\phi$ which has a `density of one particle per unit volume'. In this case, the dispersion relation~\eqref{eq:DispRelEuc} reads
\[
    \frac{m^2}{2} = 2 - s_\omega \sqrt{1-\<\mathcal{E}_{dens}[\phi] \>^2}  - s_\kappa \sqrt{1-\<\mathcal{P}_{x,dens}[\phi] \>^2},
\]
with energy density solutions 
\begin{equation}
    \label{eq:EPrelEuc}
    \<\mathcal{E}_{dens}[\phi]\>_{s_\kappa = \pm 1} = \sqrt{1-\left(2-\frac{m^2}{2} - s_\kappa \sqrt{1-\<\mathcal{P}_{x,dens}[\phi] \>^2}\right)^2}.
\end{equation}
We focus on the positive energy solutions but there is a mirror image of negative ones for the other choice of sign of the outer square root. This relation results in the energy against momentum plot in Figure~\ref{fig:DispEnMomEuc}, where the $s_\kappa = +1$ (small $\kappa \ll 1$) solution is comparable to the classical case $\<\mathcal{E}_{x,dens}[\phi] \> = \sqrt{m^2 - \<\mathcal{P}_{x,dens}[\phi] \>^2}$ shown dashed. The solutions~\eqref{eq:EPrelEuc} again display the $m, \, m' = \sqrt{8 - m^2}$ symmetry in that $\<\mathcal{E}_{dens}[\phi]\>_\pm$ for $m$ correspond to $\<\mathcal{E}_{dens}[\phi]\>_\mp$ for $m'$.

\begin{figure}[h!]
    \centering
    \includegraphics[width=1.05\textwidth]{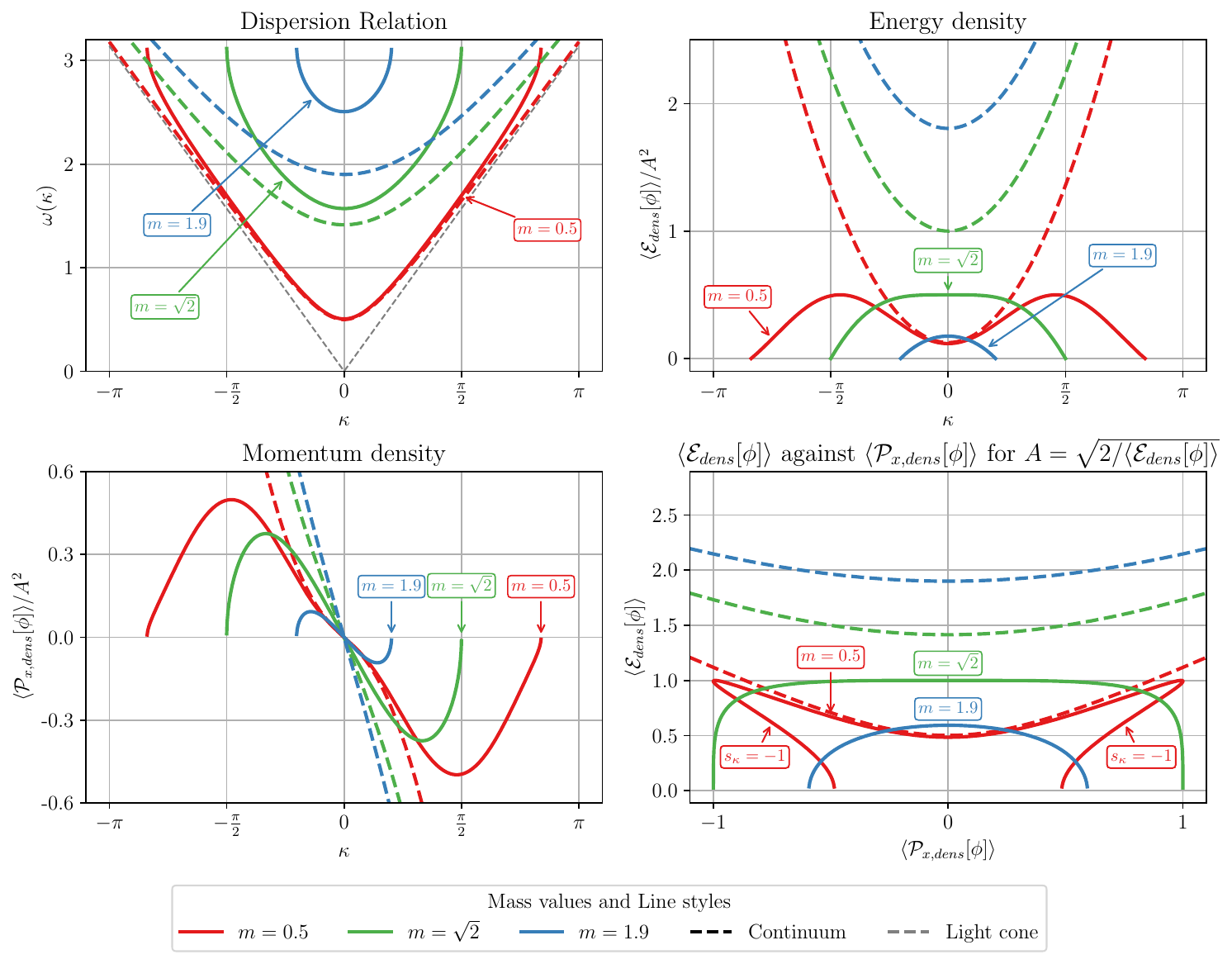}
    \caption{Dispersion relation, energy density and momentum densities in the Minkowski case for the masses $m= 0.5, \sqrt{2}, 1.9$. At lower masses, $\kappa=0$ is a local minimum of the energy density, while it is a local maximum for waves with higher masses. In the lower right, the energy density is plotted against the momentum density following equation~\eqref{eq:EPrelMin}.}
    \label{fig:DispEnMomMin}
\end{figure}

For the (1+1)-dimensional Minkowski case we now consider the plane wave solutions $\phi(t,x) = A \cos(\omega t - \kappa x)$ of the Klein-Gordon equation. In this case, the dispersion relation reads instead
\begin{equation}
    \label{eq:DispRelMin}
    m^2 = 2(\cos(\kappa)-\cos(\omega))
\end{equation}
as also found in~\cite{GrD}. This recovers the classical limit $m^2 = \omega^2-\kappa^2$ when $\omega,\kappa \ll 1$, i.e.   much smaller than the lattice spacing. We find that in this case the mass is bounded as $m < 2$ for physical states, and that the frequency behaves as
\[
    \omega(\kappa) = \arccos\left(\cos(\kappa)-{m^2\over 2}\right),
\]
with now the wavelength having to satisfy $\cos(\kappa) \in \left[-1+{m^2\over 2},1\right]$. Contrary to the Euclidean case, we find that massless modes can propagate with any wavelength $\kappa$, while the only allowed mode at $m=2$ is the constant one with $\kappa = 0$. The dispersion relation is illustrated in Figure~\ref{fig:DispEnMomMin} for $m = 0.5, \sqrt{2}, 1.9$, where it is also compared with the classical counter part. The phase velocity is again given by $v_{ph} = \omega(\kappa) / \kappa$ whereas the group velocity now reads
\[
    v_{gr} = \frac{\del \omega(\kappa)}{\del \kappa} = \frac{\sin (\kappa)}{\sqrt{1-\left(\cos (\kappa)-\frac{m^2}{2}\right)^2}},
\]
as illustrated in Figure~\ref{fig:VelMin} for the same values of $m^2$ as before, where they are also compared with the classical counter part.

The analysis of the energy and momentum densities goes as in the Euclidean case, where we now find
\begin{align*}
    \mathcal{E}_{dens}[\phi] &= \frac{A^2}{2} \left(\sin ^2(\omega)-\sin ^2(\kappa) \cos (2 t \omega-2 \kappa x)\right), \\
    \mathcal{P}_{x,dens}[\phi] &= -A^2\sin (\kappa) \sin (\omega) \sin ^2\left(\left(t-\frac{1}{2}\right) \omega-\frac{1}{2} \kappa (2 x+1)\right),
\end{align*}
and the averages 
\[
    \<\mathcal{E}_{dens}[\phi] \> = \frac{A^2}{2} \sin ^2(\omega),
\quad \quad 
    \<\mathcal{P}_{x,dens}[\phi] \> = - \frac{A^2}{2} \sin (\kappa) \sin (\omega),
\]
as before, here with a minus in the momentum density due to the metric signature. As expected, these are time independent and tend to the continuum values for small $\omega,\kappa$. They are also plotted in Figure~\ref{fig:DispEnMomMin} over the square amplitude $A^2$. We see that the profile of the energy density changes depending on whether $m< \sqrt{2}$, $m=\sqrt{2}$ or $m>\sqrt{2}$. In the first case, there are two local maxima and 3 local minima, one at $\kappa = 0$ and two at the edges of the domain. For $m>\sqrt{2}$, the point at $\kappa = 0$ is a local maximum. Comparing these with the plot of the momentum density, we find that these energy minima, either at $\kappa = 0$ or at the edges of the $\kappa$-domain, always correspond to plane waves with zero momentum density.

\begin{figure}
    \centering
    \includegraphics[width=1\textwidth]{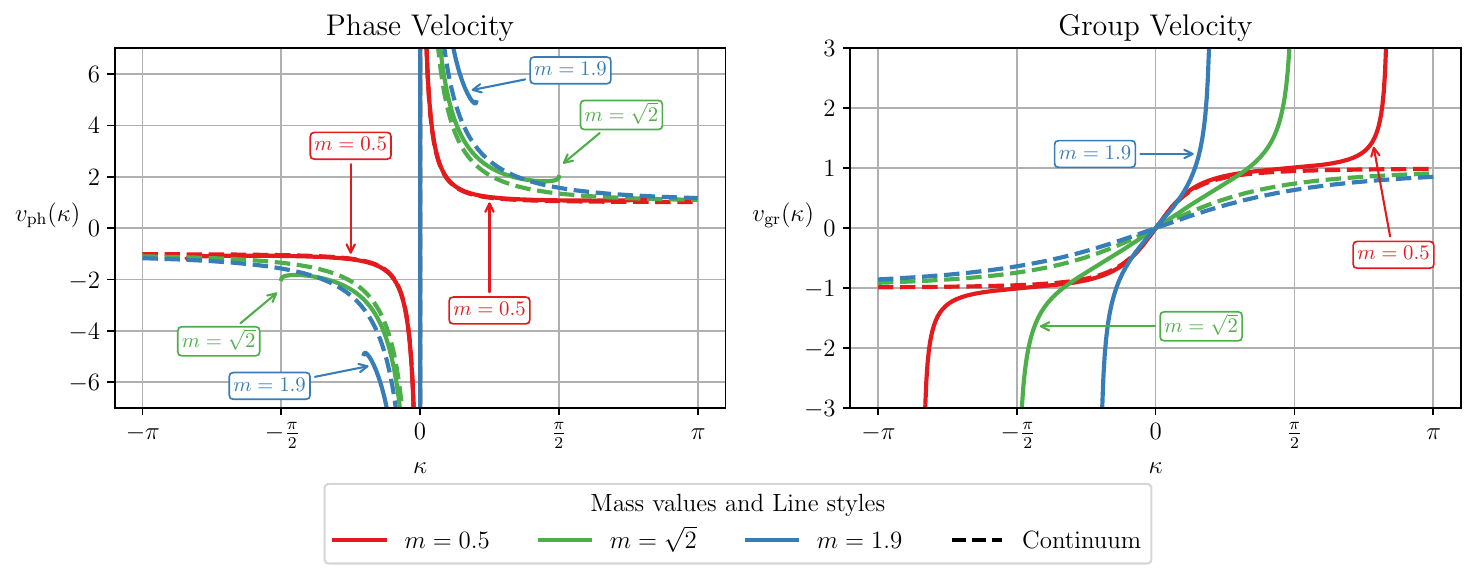}
    \caption{Phase and group velocities in the Minkowski case for $m= 0.5, \sqrt{2}, 1.9$ and comparison with their  classical counterparts shown dashed.}
    \label{fig:VelMin}
\end{figure}

As in the Euclidean case, we can relate the energy and momentum density using the dispersion relation~\eqref{eq:DispRelMin}. In this case we find
\begin{align}
    \label{eq:EPrelationMin}
    \<\mathcal{P}_{x,dens}\>^2 &= \frac{A^2}{2} \<\mathcal{E}_{dens}\> \left(1- \left(\frac{m^2}{2} + s_\omega \sqrt{1-\frac{2\<\mathcal{E}_{dens}\>}{A^2}}\right)^2\right)\\
    &\approx \left(1+\frac{m^2}{2}\right)\<\mathcal{E}_{dens}\>^2 -\left(1+{m^2 \over 4}\right) \frac{A^2}{2}\<\mathcal{E}_{dens}\> \, m^2, \nonumber
\end{align}
where the approximation is made for $\omega \ll 1$, which results in a deformation of the classical relation $\<\mathcal{E}_{dens}\>^2  - \<\mathcal{P}_{x,dens}\>^2 =  (A^2/2)\<\mathcal{E}_{dens}\> \,  m^2 $ from Example~\ref{ex:ClStressEnergy} for the Minkowski case. We can again set $A = \sqrt{2/\<\mathcal{E}_{dens}\>}$, and following similar steps to the Euclidean case, we find the energy density solutions 
\begin{equation}
    \label{eq:EPrelMin}
    \<\mathcal{E}_{dens}[\phi] \>_{s_\kappa = \pm 1} = \sqrt{1 - \left(\frac{m^2}{2} - s_\kappa \sqrt{1-\<\mathcal{P}_{x,dens}[\phi] \>^2} \right)^2}
\end{equation}
as illustrated in Figure~\ref{fig:DispEnMomMin} against their classical counterparts. We again focus on the positive energy solutions but there is a mirror image of negative energy ones. For small $\kappa$, the relevant solution that deforms the classical case is with $s_\kappa =1$ as in the introduction. In this case, we find that the configuration with vanishing momentum density $\<\mathcal{P}_{x,dens}[\phi] \> = 0$ corresponds to a local minima in the energy for $m<\sqrt{2}$ which is global for $m\le 1$, and a global maximum for $m \ge \sqrt{2}$. Note that the plots also show the $s_\kappa=-1$ solutions which are typically of lower energy and connect through zero to the negative energy solutions. These modes do not have a continuum counterpart and it is unclear as to their physical role. It is interesting that even the $s_\kappa=+1$ solutions have lowest energy at nonzero spatial momentum when $1< m<\sqrt{2}$, namely at $|\<\mathcal{P}_{x,dens}[\phi] \>| = 1$.

\section{Lagrangians on lattices with background metrics and gauge fields}\label{secO}

This section is dedicated to the computation of the EL equations and conserved charges for complex fields with a background non-Euclidean metric and/or a $U(1)$-gauge field to which the field couples. We still work on the base being an Abelian group $X$, but some results can be applied to any digraph using the quantum Riemannian geometry (QRG) formalism~\cite{BegMa} and gauge theory as in~\cite{MaSim2}. We already explained at the start of Section~\ref{secZm} how we work with the differential exterior algebra $(\Omega,\extd)$ on  $A=C(X)$ corresponding to the Cayley graph induced by a set $\CC$ of generators. We only use the group structure to provide convenient coordinates and the left invariant basis 1-forms $\{e^a\}$. 

Next, in QRG  a metric is an element $g \in \Omega^1 \tens_A\Omega^1$, with inverse metric a bimodule inner product $(\,,\,)\colon \Omega^1 \tens_A \Omega^1 \to A$. These constructions are inverse in the sense that
\[
    ((\omega,\,)\tens_A \id )g = \omega = (\id \tens_A (\, ,\omega))g,
\]
 for all $\omega\in\Omega^1$ with further reality conditions $\mathrm{flip}(*\tens *)g = g$ and $(\omega,\eta)^* = (\eta^*,\omega^*)$, where $*$ is required to extend in such a way as to commute with $\extd$ (in our case by $(e^a)^*=-e^{a^{-1}}$). Working with the left-invariant basis, we write
\[
    g = \sum_{a\in \CC} g_a e^a \tens_{C(X)} e^{a^{-1}}, \qquad 
    (e^a,e^b) = \lambda^a \delta^{a,b^{-1}},\quad \lambda^a = {1\over R_{a}g_{a^{-1}}}
\]
for non-zero real functions $g_a,\lambda^a \in C(X)$ related as shown. The lattice picture of these metric coefficients is that  $ g_{x\to xa}=g_a(x) $ is the `square length' assigned to an arrow $x\to xa$, see~\cite{MaSim2} where these notations were used for lattice gauge theory. 

Next, a bimodule connection on $\Omega^1$ in QRG is a map $\nabla\colon \Omega^1 \to \Omega^1 \tens_A \Omega^1$ together with a `generalised braiding' bimodule map $\sigma\colon \Omega^1 \tens_A \Omega^1\to \Omega^1 \tens_A \Omega^1$ obeying the left and right Leibniz rules~\cite{DVM}
\begin{equation}
    \label{eq:ConnInnerCalc}
    \nabla(\phi \omega) = \extd \phi \tens_A \omega + \phi \nabla\omega,\qquad
    \nabla(\omega\phi ) = (\nabla\omega)\phi + \sigma(\omega \tens_A\extd \phi).
\end{equation}
For an inner calculus with $\extd = [\theta, \,]$, (where $\theta = \sum_{a\in \CC} e^a$ in our case) all connections are of the form~\cite{Ma:par,BegMa}
\[
    \nabla \omega=\theta\tens_A \omega-\sigma(\omega\tens_A\theta)+\beta(\omega), 
\]
where $\sigma$ and a further bimodule map $\beta\colon\Omega^1\to \Omega^1\tens_A\Omega^1$ can be chosen freely. In our case,  being bimodule maps requires that $\sigma$ and $\beta$ have the form
\[
    \sigma(e^a\tens_A e^b) = \sum_{\substack{c\in \CC\colon \\ c^{-1}ab\in \CC}}  \sigma^{ab}{}_{c} \, e^c \tens_A e^{c^{-1}ab},\quad     
    \beta(e^a) = \sum_{\substack{b\in \CC\colon \\b^{-1}a\in \CC}}  \beta^{a}{}_{b} \, e^b \tens_A e^{b^{-1}a},
\]
for some functions $\sigma^{ab}{}_{c}, \beta^a{}_b \in C(X)$.  Given such a connection, a natural Laplacian $\Delta\colon A \to A$ (as used above for the flat metric) is given by $\Delta\coloneqq -{1\over 2}(\,,\,)\nabla \extd$ and obeys the product rule~\cite[Chap. 8]{BegMa}
\begin{equation}
    \label{eq:LaplaceProd}
    \Delta(\phi \psi) = (\Delta\phi) \psi + \phi (\Delta\psi) -\frac{1}{2} (\id + \sigma)(\extd \phi + \extd \psi)
\end{equation}
for $\phi,\psi \in A$. Here  $\Delta=\delta\extd$ where the codifferential or divergence of a 1-form is $\delta =-{1\over 2} (\,,\,) \nabla$. In both formulae, we have introduced a factor $-{1\over 2}$ compared to the canonical QRG normalisations in order to match physics conventions. An  integral $\int\colon A \to \C$ (in our case given by sum over $X$ with a measure $\mu$) is said to be divergence compatible if $\int \circ \delta = 0$. 

\subsection{Scalar field theory on $X$ with a generic metric} \label{secscalar}

We start by considering complex scalar field theory on an Abelian group $X$ endowed with a generic metric. Furthermore we choose an integration measure measure $\mu \in C(X)$ and will fix its properties as needed. In analogy to the Euclidean case, the action functional is defined as
\begin{equation*}
    S[\phi] = -{1\over 2}\int_X  \left({1\over 2}(\extd \phi,\extd \phi) + m^2 \phi^2\right)\mu \Vol 
    = {1\over 2}\int_X \left(\sum_{a\in\CC} {1\over 2}\lambda^a (\del_a \phi)^2 - m^2 \phi^2 \right)\mu\Vol
\end{equation*}
which gives the Lagrangian functional
\[
    L ={\mu\over 2} \left(\sum_{a\in\CC} {1\over 2} \lambda^a u^2_a - m^2 u^2\right).
\]
The Euler-Lagrange form and boundary form are then
\[
    EL = \left(\sum_{a\in\CC} \frac{1}{2} D_{a^{-1}} (\mu \lambda^a u_{a}) -\mu m^2 u \right) \extd_V  u \wedge \Vol,
    \qquad
    \Theta
    = \frac{1}{2} \sum_{a\in\CC} R_{a^{-1}}\left(\mu \lambda^a u_a\right)\extd_V u \wedge \Vol_a,
\]
with the equations of motion for $\phi$ therefore taking the form 
\begin{equation}
    \label{eq:eomfreescalarfield}
    {1\over{2\mu}} \sum_{a\in\CC}\del_{a^{-1}} (\mu \lambda^a \del_{a}\phi) - m^2 \phi = 0.
\end{equation}
Using the finite difference Leibniz rule and $\del_a \del_{a^{-1}} = -\del_a -\del_{a^{-1}}$, we write  this equivalently as 
\begin{equation}
    \label{eq:LaplaceMetric}
    {1\over 2}\sum_{a\in \CC}\left(\lambda^a + {R_{a}\left(\mu \lambda^{a^{-1}}\right)\over \mu}\right) \del_a \phi + m^2\phi=0.
\end{equation}
We show now that this indeed is the expected Klein-Gordon equation for the geometric Laplacian when the measure is divergence compatible.

\begin{proposition}
\label{prop:geoLap}
On an Abelian group $X$ with calculus given by $\CC\subseteq X\backslash\{e\}$ containing inverses: 

(1) the codifferential $\delta$ and geometric Laplacian given by a bimodule connection $\nabla$ do not depend on the $\beta$ part of the connection, i.e. only on the braiding $\sigma$. 

(2) In this case,   
integration with measure $\mu$ is divergence-compatible w.r.t. $\delta$  iff 
\[ R_a\left(\mu \lambda^{a^{-1}} \right)= \mu \sum_{b\in \CC}   \lambda^b \sigma^{aa^{-1}}{}_{b}\]
holds for all  $a$. 

(3) In this case, the Euler-Lagrange equation \eqref{eq:LaplaceMetric} is the Klein-Gordon equation $(\Delta+m^2)\phi=0$ for the  geometric Laplacian associated to the metric and connection. \end{proposition}

\begin{proof}
(1) Explicitly, from the form of $\sigma$, $\beta$ and $\theta$,
 \[ \nabla e^a = \sum_{b\in \CC} \left(e^b \tens_{C(X)} e^a - \sum_{\substack{c\in \CC\colon \\ c^{-1}ab\in \CC}}  \sigma^{ab}{}_{c} \, e^c \tens_{C(X)} e^{c^{-1}ab} \right) + \sum_{\substack{b\in \CC\colon \\b^{-1}a\in \CC}}  \beta^{a}{}_{b} \, e^b \tens_{C(X)} e^{b^{-1}a}.\]
 Applying the inverse metric, we see that the last term vanishes due $(e^b,e^{b^{-1}a}) = 0$ as $b\neq e$. Hence $\beta$ does not contribute to $\delta=-{1\over 2}(\ ,\ )\nabla$ and therefore to $\Delta$.
 
 (2) From the inner part of $\nabla$, we have 
 \[
    \delta e^a = -{1\over 2}\left(\lambda^{a^{-1}} - \sum_{b\in \CC}  \lambda^b\sigma^{aa^{-1}}{}_{b}\right)
 \]
 so that with measure $\mu$, 
 \begin{align*} -2\int \delta(e^a)&= \int_X \left(\lambda^{a^{-1}} - \sum_{b\in \CC}  \lambda^b\sigma^{aa^{-1}}{}_{b}\right)\mu= \sum_{g\in X} \left(\mu(g) \lambda_{g\to ga^{-1}} - \sum_{b\in \CC}  \mu(g) \lambda_{g\to gb} \sigma^{ab}{}_{c}(g)\right)\\
&= \sum_{g\in X} \left(\mu(ga) \lambda_{ga\to g} - \sum_{b\in \CC}  \mu(g) \lambda_{g\to gb} \sigma^{aa^{-1}}{}_{b}(g)\right)
 \end{align*}
 for all $g$. Vanishing of this gives the condition stated.

 (3) Using the formula for $\delta$, we compute
\begin{align*} 
    (\ ,\ )\nabla\extd \phi 
    &= \sum_{a\in \CC}(\ ,\ )\nabla(\del_a \phi e^a)
    = \sum_a\left((\extd \del_a \phi, e^a) + \del_a \phi \delta e^a\right)\\
    &= \sum_{a\in \CC}\left( \lambda^a \del_a \del_{a^{-1}}\phi + \left(\lambda^{a^{-1}} - \sum_{b\in \CC} \lambda^b \sigma^{aa}{}_b\right)\del_a \phi\right)\\
    &= - \sum_{a\in \CC}\left( \lambda^a + {R_a\left(\mu \lambda^{a^{-1}}\right)\over \mu} \right)\del_a \phi
\end{align*}
which recovers -2 times the first term of equation~\eqref{eq:LaplaceMetric}.
\end{proof}

For a connected calculus, the condition in part (2) uniquely determines $\mu$, if it exists, up to an overall constant from the metric coefficients $\lambda_a$ and the braiding $\sigma$. This is a graph version of the way the Riemannian measure of integration is determined by the metric in Riemannian geometry.

To end this section, we note that while we have not explicitly developed the variational calculus formalism and hence the Euler-Lagrange equations for a general graph, the above easily extrapolates to this. We consider the bundle $X\times \R \to X$ for $X$ the set of vertices and fields $\phi\in C(X)$.  For a general directed graph, $\Omega^1_X$ is spanned by the arrows $e^{x\to y}$ and has a specific module structures as in~\cite{BegMa} where we multiply from the left by the value at $x$ and from the right by the value at $y$ (on in the case of a group, the  left-invariant basis is $e^a = \sum_{g\in X} e^{g\to ga}$). The differential of a function is $\extd \phi=\sum_{x\to y}( f(y)-f(x))e^{x\to y}$ (which likewise implies the form we used when $X$ is a group).  The inverse metric takes the form~\cite{BegMa}
\[  (e^{x\to y}, e^{y'\to x'})= \lambda_{x\to y} \delta_{x,x'}\delta_{y,y'}\]
and the bimodule connections are again given by bimodule maps $\sigma,\beta$ now decomposed as 
\[
\sigma(e^{x\to y} \tens_{\C(X)} e^{y\to z}) = \sum_{y':x\to y'\to z}\sigma^{xz,y}{}_{y'} e^{x\to y'} \tens_{\C(X)} e^{y'\to z},
\]
\[
\beta(e^{x\to y})=\sum_{w: x\to w\to y}\beta^{xy}{}_{w} e^{x\to w}\tens_{\C(X)} e^{w\to y}.
\]
Similarly, the Euler-Lagrange equations of motion \eqref{eq:LaplaceMetric}, when evaluated at $x$, can immediately be written in the more general form
\begin{equation}
    \label{eq:LaplaceGraph}
{1\over 2}\sum_{y:x\to y}  \left(\lambda_{x\to y}+ \frac{\mu_y \lambda_{y\to x}}{\mu_x }\right) (\phi_y - \phi_x)+ m^2\phi_x, 
\end{equation}
for all $x$, which makes sense for any graph. We also have:
\begin{lemma} Equation~\eqref{eq:LaplaceGraph} extremises the action
\[ S[\phi]=-{1\over 2}\int_X\left({1\over 2} (\extd\phi,\extd\phi)+ m^2 \phi^2\right) \]
regarded as a quadratic function of the $|X|$ variables $\{\phi_x\}$ and understanding $\int_X$ as a sum over the vertices with weight $\mu$. \end{lemma}
\begin{proof} Unpacking the definitions, we obtain
\[ S[\phi]={1\over 4}\sum_{x\to y}\mu(x)\lambda_{x\to y}(\phi(x)-\phi(y))^2-{1\over 2}\sum_{x\in X}\mu(x) m^2 \phi(x)^2. \]
Then 
\[ {\del S\over\del \phi(z)}= {1\over 2}\sum_{y:z\to y}\mu(z)\lambda_{z\to y}(\phi(z)-\phi(y))+ {1\over 2}\sum_{x:x\to z}\mu(x)\lambda_{x\to z}(\phi(z)-\phi(x))-m^2\mu(z)\phi(z)\]
where we obtain terms involving $\phi(z)$ when $x=z$ or when $y=z$. Factoring out $\mu(z)\ne 0$, we see that ${\del S\over\del\phi(z)}=0$ amounts to~\eqref{eq:LaplaceGraph} on a change of variables. \end{proof}

This implies that in the group case, if we just wanted the EL equations then this naive method gives the same answer without the full variational calculus and jet formalism, at least for a free scalar field action. Then, in a similar way to Proposition~\ref{prop:geoLap}, we have: 

\begin{proposition}
    \label{prop:geoLapGraph}
In a bidirected graph $X$:

(1) the codifferential and geometric Laplacian given by a bimodule connection $\nabla$ do not depend on the $\beta$ part of the connection, i.e. only on the braiding $\sigma$. 

(2) In this case,   
integration with measure $\mu$ is divergence-compatible wrt $\delta=(\ ,\ )\nabla$, i.e. $\int\circ\delta=0$, iff 
\[ \mu_y \lambda^{y\to x} = \mu_x\sum_{y':x\to y'} \lambda^{x\to y'} \sigma^{xx,y}{}_{y'}\]
holds for all  $y\to x$. 

(3) In this case,~\eqref{eq:LaplaceGraph} is the Klein-Gordon equation for the geometric Laplacian associated to the graph metric.
\end{proposition}

The proof is omitted as it is analogous to the group case when this is expressed in graph terms.  The relation in (2) from Proposition~\ref{prop:geoLapGraph} is remarkably the same as found in~\cite{BegMa1} in the context of quantum geodesic flows on graphs. 

\subsection{Conserved charge associated to global $U(1)$ symmetry}

Another example of symmetries that we have not considered so far with the base a lattice or a discrete Abelian group are vertical continuous symmetries as in Example~\ref{ex:ClU1} in the continuum case. Here we can copy everything from the continuum, as the symmetry only acts on the vertical 1-forms which form a Grassmann algebra. We consider a complex scalar field on a discrete Abelian group $X$ with inverse metric given by $\lambda^a$, modelled by the bundle $X\times \C\to X$, meaning that the jet bundle will now have coordinates $\bar u,u$ and the respective derivatives. We choose the action
\begin{align*}
    S[\phi,\phi^*] = -{1\over 2}\int_X  \left({1\over 2}(\extd \phi^*,\extd \phi) +  m^2 |\phi|^2\right)\mu\Vol 
    = {1\over 2}\int_X\left(\sum_{a\in\CC} {1\over 2}\lambda^a (\del_a \phi^*)\del_a \phi -m^2 |\phi|^2 \right)\mu \Vol
\end{align*}
with the Lagrangian $L = \left(\sum_a {1\over4}\lambda^a \bar u_a u_a - {m^2\over2} \bar u u\right)$. This system has a global $U(1)$-symmetry $\phi \mapsto e^{\imath \varphi} \phi$, $\phi^* \mapsto e^{-\imath \varphi} \phi^*$, which can be implemented by setting $\iota_V \extd_V u_I = \imath \varphi u_I$ and $\iota_V \extd_V \bar u_I = -\imath \varphi \bar u_I$, with $\iota_H = 0$. The EL and boundary form are
\[
    EL = \left(\left(\sum_{a\in\CC} \frac{1}{4} D_{a^{-1}} (\mu \lambda^a u_{a}) -{\mu m^2\over 2} u \right) \extd_V  \bar u+ \left(\sum_{a\in\CC} \frac{1}{4} D_{a^{-1}} (\mu \lambda^a \bar u_{a}) -{\mu m^2\over 2} \bar u \right) \extd_V   u\right)\wedge \Vol,
\]
\[
    \Theta
    = -\frac{1}{4} \sum_{a\in\CC} R_{a^{-1}}\left(\mu \lambda^a\right) \left(\bar u_{a^{-1}} \extd_V u + u_{a^{-1}} \extd_V \bar u \right) \wedge \Vol_a,
\]
The EL form leads to the equations of motion \eqref{eq:eomfreescalarfield} for $\phi$ and $\phi^*$. Just as in Example~\ref{ex:ClU1} we have $\sigma = 0$, and the conserved Noether current is
\[
    j_{U(1)} =-\iota_V \Theta = \frac{\imath  \varphi }{4} \sum_{a\in\CC} R_{a^{-1}}\left(\mu \lambda^a\right) \left( \bar u_{a^{-1}} u -  u_{a^{-1}} \bar u\right) \Vol_a,
\]
and in terms of the fields:
\[
    j_{U(1)}[\phi,\phi^*] = \frac{\imath  \varphi }{4} \sum_{a\in\CC} R_{a^{-1}}\left(\mu \lambda^a\right) \left( (\del_{a^{-1}}\phi^*) \phi -  (\del_{a^{-1}}\phi) \phi^*\right) \Vol_a.
\]
To be thorough we check that this quantity is conserved explicitly:
\begin{align*}
    &\extd_X j_{U(1)}[\phi,\phi^*] = -\frac{\imath  \varphi }{4} \sum_{a\in \CC} \del_a \left( R_{a^{-1}}\left(\mu \lambda^a \del_a \phi^*\right)  \phi -  R_{a^{-1}}\left(\mu \lambda^a \del_{a}\phi \right) \phi^*\right) \Vol\\
    &= -\frac{\imath  \varphi }{4} \sum_{a\in \CC} \left( - \del_{a^{-1}}\left(\mu \lambda^a \del_a \phi^*\right)  R_a \phi + R_{a^{-1}}\left(\mu \lambda^a \del_a \phi^*\right) \del_a \phi + \del_{a^{-1}}\left(\mu \lambda^a \del_{a}\phi \right) R_{a}\phi^* - R_{a^{-1}}\left(\mu \lambda^a \del_{a}\phi \right) \del_a\phi^*\right)\\
    &= -\frac{\imath  \varphi }{4} \sum_{a\in \CC} \left( - 2\mu m^2 \phi^*  R_a \phi + \left(\mu \lambda^a \del_a \phi^* + 2\mu m^2 \phi^* \right) \del_a \phi + 2 \mu m^2\phi R_{a}\phi^* - \left(\mu \lambda^a \del_{a}\phi + 2\mu m^2 \phi\right) \del_a\phi^*\right)\\
    &= 0
\end{align*}
where we have used $\del_a R_{a^{-1}} = - \del_{a^{-1}}$, that $R_{a^{-1}} = \del_a + \id$ and the EL equations for $\phi, \phi^*$.

If there is a preferred time direction then we write the lattice as $\Z \times \Z^{m-1}$. From Corollary~\ref{corQ}, there is a conserved charge $Q_{U(1)}[\phi,\phi^*]$ associated to $j_{U(1)}$ and given by 
\begin{align*}
    Q_{U(1)}[\phi,\phi^*] = {\imath \varphi \over 4}\int_{\Z^{m-1}} &\left[\left(R_{t^{-1}}(\mu \lambda^t) + \mu \lambda^{t^{-1}}\right) ((\del_{t^{-1}}\phi^*) \phi - (\del_{t^{-1}}\phi)\phi^*)\right.\\
    &\left.+\mu \lambda^{t^{-1}}((\del_{t^{-1}}\phi^*) \del_t \phi - (\del_{t^{-1}}\phi)\del_t\phi^*)\right]\Vol_{\Z^{m-1}}.
\end{align*}  
From this, we can read the $U(1)$-charge density $\rho$ and the current density as having the form
\begin{align*}
    J^x_{U(1)}[\phi,\phi^*] = {\imath \varphi \over 4} &\left[\left(R_{x^{-1}}(\mu \lambda^x) + \mu \lambda^{x^{-1}}\right) ((\del_{x^{-1}}\phi^*) \phi - (\del_{x^{-1}}\phi)\phi^*)\right.\\
    &\left.+\mu \lambda^{x^{-1}}((\del_{x^{-1}}\phi^*) \del_x \phi - (\del_{x^{-1}}\phi)\del_x\phi^*)\right].
\end{align*}  
In the continuum limit, we expect $\del_{t^{\pm 1}}\phi \mapsto \pm \del_0 \phi$, $\del_{x^{\pm 1}}\phi \mapsto \pm \del_i \phi$ and similar for $\phi^*$, that $R_{t^{-1}}\mapsto \id$ and $\lambda^{t^{\pm 1}} \mapsto \cg^{00}$, $\lambda^{x^{\pm 1}} \mapsto \cg^{ii}$. Therefore in both $Q_{U(1)}[\phi,\phi^*]$ and $J^x_{U(1)}[\phi,\phi^*]$, the second line vanishes and the first recovers the expression from Example~\ref{ex:ClU1} for $g$ the Euclidean or Minkowski metrics. 

For $\lambda$ the Euclidean or Minkowski metric with $\lambda^{t} = \lambda^{t^{-1}} = 1$, $\lambda^{x} = \lambda^{x^{-1}} = \pm 1$, we can redo the calculations for the conserved current and stress-energy tensor as in Proposition~\ref{prop:ConservedCurrentScalarField}, now taking the doubling of the fields into account to find
\[
    j = - \sum_{a,b\in \CC} \epsilon^b \left(\frac{\lambda^a}{4}\left(\bar u_{a^{-1}}+\frac{1}{2} \bar u_{b a^{-1}}\right)u_b+ \frac{\lambda^a}{4}\left(u_{a^{-1}}+\frac{1}{2}u_{b a^{-1}}\right) \bar u_b + \delta^{a}_b L \right) \Vol_a, 
\]
\[
    T_{ab}= -\left(\frac{1}{4} \bar u_a u_b + \frac{1}{4} u_a \bar u_b +\frac{1}{8} \bar u_{ab} u_b +\frac{1}{8} u_{ab} \bar u_b + \frac{\delta^{a^{-1}}_b}{\lambda^a} L \right).
\]
In the case of a plane wave solution $\phi(t,x) = A e^{-\imath (\omega t + \lambda^x \kappa x )}$ in (1+1)-dimensions these result in the energy and momentum density
\[
    \mathcal{E}_{dens}[\phi,\phi^*] = |A|^2 \sin ^2(\omega),
\quad \quad 
    \mathcal{P}_{x,dens}[\phi,\phi^*] = \lambda^x |A|^2 \sin (\kappa) \sin (\omega).
\]
with now $A\in \C$, which as in continuum field theory correspond to 2 times the spatial averages in equation~\eqref{eq:RealScalarEPdens} of the real scalar field, since we now have two fields $\phi,\phi^*$. Again, we find similar relations as in equations~\eqref{eq:EPrelationEuc},~\eqref{eq:EPrelEuc}, ~\eqref{eq:EPrelationMin} and ~\eqref{eq:EPrelMin}, but now for the energy and momentum densities directly instead of for their averages.

The $U(1)$-charge and current densities for these solutions read
\[
    \rho_{U(1)}[\phi,\phi^*] = |A|^2 \varphi \sin(\omega) \cos(\omega), 
    \quad \quad
    J^x_{U(1)}[\phi,\phi^*] = |A|^2 \varphi  \sin(\kappa) \cos(\kappa),
\]
which recover the classical values in Example~\ref{ex:ClU1} for $\omega,\kappa \ll 1$.

\subsection{Complex scalar field theory on $X$ with background gauge field}
\label{sec:U1symm}

We now extend the previous case to include a noncommutative $U(1)$ background gauge field $\alpha$. We take this in the simplest form where $\alpha$ is an anti-Hermitian 1-form $\alpha = \sum_{a\in \CC} \alpha_a e^a \in \Omega^1$. Here $\alpha^* = -\alpha$ means  $(\alpha_a)^* = R_a\alpha_{a^{-1}}$ for the coefficients in our left-invariant basis. The discussion in this section can be generalised to a $U(d)$-connection or a connection with values in $\C G^+$ for a gauge group $G$, i.e. to the setting discussed in~\cite{MaSim2}. 

The curvature of $\alpha$ is $F = \extd \alpha + \alpha \wedge \alpha \in \Omega^2$, where even in the Abelian case $\alpha\wedge \alpha \neq 0$ due to noncommutative effects. The gauge field induces a bimodule connection $\nabla_\alpha$ on $C(X)$, 
\[
    \nabla_\alpha\colon C(X) \to \Omega^1 \tens_{C(X)} C(X) \simeq \Omega^1, 
    \qquad \nabla_\alpha \phi = \extd \phi - \phi \alpha.
\]
Note that we looked at bimodule connections on $\Omega^1$, but there is an analogous definition for any $A$-bimodule $E$,  now with $\sigma_E: E\tens_A\Omega^1\to \Omega^1\tens_A E$. In our case, $A=E=C(X)$ and the braiding $\sigma_E$ is
\[ 
\sigma_\alpha\colon \Omega^1 \simeq C(X)\tens_{C(X)}  \Omega^1 \to \Omega^1 \tens_{C(X)} C(X) \simeq \Omega^1,
\qquad \sigma_\alpha(\omega) = \sum_{a\in \CC} (1+\alpha_a)\omega_a e^a,
\]
for $\omega = \sum_{a\in \CC} \omega_a e^a$. The braiding here was computed through the right Leibniz rule for $\phi, \psi$, namely 
\[
    \sigma_\alpha(\phi\extd \psi) = \nabla_\alpha (\phi\psi) - (\nabla_\alpha \phi)\psi = \sum_{a\in \CC} (1+\alpha_a)\phi(\del_a\psi) e^a.
\]

Next, gauge transformations are unitary elements $\gamma \in \C(X)$, $\gamma \gamma^* =\gamma^* \gamma = 1$, which act as
\[
\alpha^\gamma = \gamma \alpha \gamma^* + \gamma \extd \gamma^*,
\quad \quad 
F^\gamma = \gamma F\gamma^*,
\quad \quad
\phi^\gamma = \phi \gamma^*,
\quad \quad 
\nabla_{\alpha^\gamma} \phi^\gamma = (\nabla_\alpha \phi)\gamma^*
\] 
where $\nabla_{\alpha^\gamma}$ is the connection as constructed from $\alpha^\gamma$. The action functional in this case can be defined as 
\begin{align*}
S[\phi,\phi^*] &= \int_X \mu \left(-\frac{1}{4}((\nabla_A \phi)^*, (\nabla_A \phi)) - \frac{m^2}{2} |\phi|^2 \right)\\
&= \int_X \mu \left(-\frac{1}{4}\sum_{a\in \CC}\lambda^a(\nabla_\alpha \phi)_a^* R_{a}(\nabla_\alpha \phi)_{a^{-1}} - \frac{m^2}{2} |\phi|^2 \right)
\end{align*}
where we write $\nabla_\alpha \phi = (\nabla_\alpha \phi)_a e^a$ with $(\nabla_\alpha \phi)_a = \del_a \phi -\phi \alpha_a$. Similarly for $(\nabla_\alpha \phi)^*$, but here we have
\[
   (\nabla_\alpha \phi)^* = \extd \phi^* +\alpha \phi^* = \sum_{a\in \CC}(\del_a \phi^* + \alpha_a R_a \phi^*) e^a
\]
so that $(\nabla_\alpha \phi)^*_a = (\del_a \phi^* + \alpha_a R_a \phi^*)$.
This results in the Lagrangian
\[
    L(u,\bar u,u_a,\bar u_a) 
    = \mu \left(\sum_{a \in \CC}{\lambda^a\over 4} (\bar u_a + \alpha_a (\bar u + \bar u_a)) (u_{a} + R_a(\alpha_{a^{-1}}) (u + u_a))  - {m^2\over2} \bar u u \right)
\]
Computing the EL form and simplifying leads to 
\begin{align*}
    EL = -&\sum_{a\in\CC}\left({\mu m^2 \over 2}u + {\mu \over 4}\left(\lambda^a + {R_a(\mu \lambda^{a^{-1}})\over \mu} \right) \left(1+R_a \alpha_{a^{-1}}\right) u_a\right.\\
    + &\left.{\mu \over 4}\left(\lambda^{a^{-1}} R_{a^{-1}}\alpha_a + {R_a(\mu \lambda^{a^{-1}})\over \mu} \left(\alpha_a - \alpha_a R_a\alpha_{a^{-1}}\right) \right)u\right)\extd_V\bar u\wedge \Vol\\
    -&\sum_{a\in\CC}\left({\mu m^2 \over 2}\bar u + {\mu \over 4} \left(\lambda^a + {R_a(\mu \lambda^{a^{-1}})\over \mu} \right) \left(1+ \alpha_{a}\right) \bar u_a\right.\\
    + &\left.{\mu \over 4}\left(\lambda^{a^{-1}} \alpha_{a^{-1}} + {R_a(\mu \lambda^{a^{-1}})\over \mu} \left(R_a \alpha_{a^{-1}} -\alpha_{a} R_a \alpha_{a^{-1}}\right) \right)\bar u\right)\extd_V u\wedge \Vol
\end{align*}
which corresponds to the following EL equation for $\phi$
\begin{align}
    \label{eq:EOMphiconn}
    &{1 \over 2}\sum_{a\in\CC}\left(\left(\lambda^a + {R_a(\mu \lambda^{a^{-1}})\over \mu} \right) \left(1+R_a \alpha_{a^{-1}}\right) \del_a \phi\right.\\\nonumber
    &\quad +\left. \left(\lambda^{a^{-1}} R_{a^{-1}}\alpha_a + {R_a(\mu \lambda^{a^{-1}})\over \mu} \left(\alpha_a - \alpha_a R_a\alpha_{a^{-1}}\right) \right)\phi\right)+ m^2\phi=0
\end{align}
and similarly for $\phi^*$. Note tha this reduces to \eqref{eq:eomfreescalarfield} for $\alpha = 0$ as expected. The boundary form is
\begin{align*}
    \Theta = &-\sum_{a\in\CC}\left(R_{a^{-1}}\left({\mu \lambda^a \over 4}(1+\alpha_a)\right)(u_{a^{-1}}-\alpha_{a^{-1}}u) \extd_V \bar u \right.\\
    &+\left.R_{a^{-1}}\left({\mu \lambda^a \over 4}(1+R_a \alpha_{a^{-1}})\right)(u_{a^{-1}}-(R_{a^{-1}}\alpha_a) u) \extd_V u \right)\wedge \Vol_a
\end{align*}

We now obtain a geometric interpretation of along the same lines as Proposition~\ref{prop:geoLap}.  Given a bimodule connection $\nabla$ on $\Omega^1$ with braiding $\sigma$ on $\Omega^1$, the connection 1-form $\alpha$ induces a new connection on $\Omega^1$ (a gauged version of it), namely
\[
    \nabla_{\Omega^1,\alpha} \omega = \nabla \omega - \sigma(\omega \tens_{C(X)} \alpha),
    \quad \quad
    \sigma_{\Omega^1,\alpha}(\omega\tens_{C(X)} \eta) = \sigma(\omega \tens_{C(X)} \sigma_\alpha(\eta)).
\]
(This is the tensor product bimodule connection~\cite{BegMa} on $\Omega^1 \tens_{C(X)} E \simeq \Omega^1$ as constructed from $\nabla, \nabla_\alpha$.) The geometric Laplacian in this setting is defined as $\Delta_\alpha =-{1\over 2} (\, , \,) \nabla_{\Omega^1,\alpha} \nabla_\alpha$, which can be expanded to give
\[
\Delta_\alpha \phi = \Delta \phi +{1\over 2} (\,,\,)(\id + \sigma) (\extd \phi \tens_{C(X)} \alpha) - \phi \delta \alpha -{1\over 2} \phi (\,,\,)\sigma(\alpha \tens_{C(X)} \alpha).
\]
Computing this expression and assuming the same relation between the measure $\mu$, inverse metric $\lambda^a$ and braiding $\sigma$ as in Proposition~\ref{prop:geoLap}, we recover exactly the operator in \eqref{eq:EOMphiconn}, so that the latter becomes the covariantised Klein-Gordon equation $(\Delta_\alpha+m^2)\phi=0$. We have done this for $X$ an Abelian group but as with the end of Section~\ref{secscalar}, there is also a version on any graph.

\section{Concluding remarks}\label{secrem}

We have presented what we believe to be the first derivation of field equations and  exactly conserved Noether charges for classical mechanics and field theory on a lattice from an action. We viewed the lattice as an exact noncommutative geometry and used a theory of jet bundles from our previous work~\cite{MaSim1}, albeit an easy case as the base coordinate algebra remains commutative and only differential forms (given by arrows of the lattice) non-commute with functions. The resulting conserved quantities are not obvious even for a lattice line (but can be easily  checked in this case). Having a good foundation for classical field theory including conserved quantities is also an important step for the Hamiltonian form of quantum field theory as a complement to existing functional integral methods even on a lattice~\cite{Smi}. 

There are several further directions to be explored in further work. First of all, as suggested by Proposition~\ref{prop:geoLapGraph}, our methods should extend to any graph, not necessarily an Abelian group lattice (which is the case covered here). The explicit approach to jet bundles in~\cite{MaSim1} requires a suitable torsion free flat background connection $\nabla$, which can be solved even for some non-Abelian group Cayley graphs, but can also be avoided altogether when coming form more abstract methods as in the subsequent work~\cite{Flo}. It should also be straightforward to cover the hybrid case of $\R\times G$ where $\R$ is a continuous time and $G$ is a discrete group or graph. 

Another issue that came up already for the lattice is what exactly is the right notion of symmetry for Noether's theorem. The classical picture involves interior products by vector fields, which in noncommutative geometry is additional data (i.e. not determined by the exterior algebra alone) and as a result our conserved current involved an additional term needed for the conservation. The origin of this and the appropriate generalisation in terms of symmetries or `Killing vectors' of a curved quantum metric should be explored further. It would also be of interest to find the conserved stress tensor for gauge theory even on a lattice in the formalism of~\cite{MaSim2}, and eventually to find a natural Einstein tensor which is currently lacking in any generality. 

We also limited ourselves to trivial bundles. The jet formalism itself and the Anderson-Zuckerman approach~\cite{And,Zuc}  can include nontrivial bundles $E$ as (finitely generated projective) bimodules over the coordinate algebra $A$ and our approach on the lattice should extend to this with more care.  Finally, the current methodology (even for trivial bundles) works similarly for noncommutative $A$ provided there are nice properties for the exterior algebra, such as a global basis. The simplest case of $A$ noncommutative but with a central Grassmann algebra basis is straightforward and would apply to $A$ the fuzzy sphere, for example. Other examples where there is no problem to define the jet bundle~\cite{MaSim1} and where our methods should work would be to variational calculus and hence classical field theory on the Majid-Ruegg $\kappa$-Minkowski spacetime. 

Finally,  while we found the Anderson-Zuckerman approach particularly natural for noncommutative geometry, there are other variants and routes to physics in the case of a classical manifold, e.g.~\cite{For, MusHro} which may have aspects that extend to the noncommutative case. At the mathematical level, another approach to jet bundles is as Hopf algebroids, see~\cite{HanMa} where examples are obtained by cotwist.  Note that knowing the noncommutative analogue of sections of the jet bundle as in~\cite{MaSim1,Flo, HanMa} is not enough to know what is the noncommutative analogue of $C(J^\infty)$ and its differential calculus.  Also, jets are dual to differential operators, see~\cite{KraVer}, with the latter the natural coordinate-invariant notion of the Heisenberg algebra for a classical manifold. This could give another route to the physics possibly bypassing jets entirely. Here, a Hopf algebroid of differential operators on a possibly noncommutative algebra $A$ with suitable calculus $(\Omega,\extd)$ was found in~\cite{Gho}, and examples by twisting the classical algebra of differential operators in~\cite{Xu}.

\end{document}